\documentclass[journal,twocolumn]{IEEEtran}
\usepackage{bypass}
\usepackage{paper}
\usepackage{subfigure}
\usepackage{graphicx}
\usepackage{algorithmic}
\usepackage{algorithm}
\usepackage{comment}
\usepackage{multirow}
\usepackage{amssymb,amsmath}
\usepackage{epsfig}
\usepackage{cite}
\newtheorem{corollary}{Corollary}
\newtheorem{theorem}{Theorem}
\newtheorem{lemma}{Lemma}
\newtheorem{observation}{Observation}

\newtheorem{definition}{Definition}

\begin{document}

\title{Maximizing Reliability in WDM Networks through Lightpath Routing}

\author{Hyang-Won Lee,~\IEEEmembership{Member,~IEEE,} Kayi Lee,~\IEEEmembership{Member,~IEEE,}
Eytan Modiano,~\IEEEmembership{Fellow,~IEEE}
\thanks{Hyang-Won Lee is with Konkuk University, Seoul, Republic of Korea.  Kayi Lee is with Google, Inc., Cambridge, MA. Eytan Modiano is with the Massachusetts Institute of Technology, Cambridge, MA 02139 (e-mail: \{hwlee, kylee, modiano\}@mit.edu).}
\thanks{This work was supported by NSF grants CNS-0830961 and CNS-1017800, and by DTRA grants HDTRA1-07-1-0004 and HDTRA-09-1-0050.}}

\maketitle

\begin{abstract}
We study the reliability maximization problem in WDM networks with random link failures. Reliability in these networks is defined as the probability that the logical network is connected, and it is determined by the underlying lightpath routing, network topologies and the link failure probability. By introducing the notion of lexicographical ordering for lightpath routings, we characterize precise optimization criteria for maximum reliability in the low failure probability regime. Based on the optimization criteria, we develop lightpath routing algorithms that maximize the reliability, and logical topology augmentation algorithms for further improving reliability. We also study the reliability maximization problem in the high failure probability regime.



\end{abstract}


\section{Introduction}\label{sec:intro}
Modern communication networks are constructed using a layered approach, with one or more electronic layers (e.g., IP, ATM, SONET) built on top of an optical fiber network. The survivability of such networks under fiber failures largely depends on how the logical electronic topology is embedded onto the physical fiber topology. In the context of WDM networks, this is known as \emph{lightpath routing}. However, finding a reliable lightpath routing is rather challenging because it must take into account the sharing of physical fibers by logical links and its impact on the connectivity of the logical topology. Hence, the survivability of a layered network is a complex function of logical topology, physical topology, lightpath routing, and link failure probability. In this paper, we study reliable layered network design assuming that physical links fail at random with some probability, where multiple links may fail simultaneously.

The probabilistic failure model represents a snapshot of a network where links fail and are repaired after a certain time as in many practical scenarios \cite{zhang:review}. Hence, the link failure probability can be viewed as the average fraction of time that a link is in a failed state. This random failure model is somewhat general in that it can be used to model both networks with rare link failures as well as more frequent failures. It thus enables thorough understanding of network survivability in various failure regimes. For this reason, several works in the literature study survivable network design under the random failure model \cite{zhang:review, zhang:service, tornatore:availability, segovia:topology}.



In the context of layered networks with random physical link failures, a natural survivability metric is the probability that given a lightpath routing, the logical topology remains connected; we call this probability the \emph{cross-layer (network) reliability}. The cross-layer reliability reflects the survivability performance achieved by the layered network. Hence, it is desirable to design a layered network that maximizes the reliability.  {\bf Although the single-layer network design problem has been extensively studied \cite{boesch:existence, wang:proof, myrvold:uniformly, bauer:validity, myrvold:reliable,ath:some, ath:counterexamples, boesch:unreliability}, the layered network reliability problem remains largely unexplored. Existing work in the area~\cite{kurant:survey, ArmitageINFOCOM97, CrochatJSAC98, sasakichingfong02, kurant:survivable, SubramaniamChoiSurvivableEmbedding, modiano02survivable, binhaoISCC, binhaoHPSR03, ramamurthyBroadnets04} has mostly focused on finding a lightpath routing that survives a single physical link failure, rather than  finding the one with maximum reliability. 
Our work in~\cite{lee:crosslayer} was the first study to {\it maximize} the tolerance of such physical failures for a lightpath routing, and cross-layer reliability
was introduced in~\cite{lee:reliability} to generalize this notion. In particular, we extended the polynomial expression for single-layer network reliability to the layered setting, and developed approximation algorithms for reliability computation. We also demonstrated a positive correlation between the reliability and Min Cross Layer Cut (MCLC; The precise definition of MCLC is presented in Section \ref{sec:model}) 
in the low failure probability regime, and experimented with MCLC as the objective in our lightpath routing algorithm to approximate reliability maximization. 


Our goal is to fully characterize the structures that contribute to the reliability in a layered network. This gives us the precise optimization criterion for maximizing the reliability. Although optimizing the exact criterion is infeasible in practice, the insight allows us to develop a new objective that better approximates reliability maximization.
}


Typically, real-world networks experience very low link failure probabilities, and are designed accordingly. For instance, the failure probability of a 1000-mile cable in the Bellcore network is estimated to be about 0.006 \cite{to:unavailability}. However, in recent years there has been an increased concern about the impact of natural disasters or physical attacks on network survivability.    Natural disasters, such as earthquakes and hurricanes or floods can lead to a large number of (possibly localized) link failures that cannot be survived by networks designed to deal only with isolated failures \cite{Wu09,Neu11}.  Worse yet, a physical attack on the network by weapons of mass destruction, such as an Electromagnetic Pulse (EMP), can lead to widespread failures throughout large geographical areas \cite{Neu11, Fos04, Wil04}.   Such an attack can have a disastrous effect on telecommunication links that rely on electronic components from fiber amplifiers to regenerators, switches and routers for their operation.  Worse yet, such an attack is likely to disrupt the power grid \cite{Rad07,Fed10}, which can in-turn lead to significant additional  (cascading) failures of communication links, as was recently observed during a blackout event in Italy \cite{rosato2008}.  Thus, while typically one may expect extremely low failure probabilities, and design networks accordingly, such designs may not be robust to widespread failures that may result from a natural disaster or attack.  Furthermore, it may be worthwhile to strengthen networks of critical importance so that they can withstand such scenarios.  

Our primary focus in this work is on the low failure probability regime, as that is the regime that networks are typically designed for.  However,  to account for the increasing concerns with large scale failures, we also characterize network survivability in higher failure probability regimes.  While such designs may not be applicable to most networks, they may prove valuable to the design of networks with stringent survivability requirements.  

{\bf
One of the major challenges in the area of cross-layer survivability is the inherent complexity of the problems.  For example, 
in~\cite{lee:crosslayer}, we proved that the MCLC, a critical component in 
layered network reliability, is NP-hard to compute and approximate with within a $O(\log n)$ factor. Therefore, problems for maximizing cross-layer reliability is likely to be intractable.  
The common approach in existing lightpath routing algorithms involves finding the physical routes of all logical links jointly, typically by solving an ILP 
that captures the routing decision of all the logical links, which is often infeasible for large networks. In this paper, we consider
a different approach by incrementally improving the layered network, one logical link at a time. Such an approach has the advantages 
over the existing algorithms:
\begin{enumerate}
\item{Scalability:} Routing the logical links incrementally
reduces the problem space significantly. As a result, it is more applicable to large networks. 
\item{Solution Quality:} The incremental approach allows us to use a more sophisticated objective function that better approximates the cross-layer reliability. As a result, the lightpath routings given by the new algorithm result in much higher reliability than existing algorithms.
\end{enumerate}
We also apply a similar idea to a different setting where the logical topology can be augmented to improve reliability. We develop an augmentation algorithm to find
a good placement of a new logical link, and observe that reliability can be improved significantly, especially when the augmentation increases the MCLC.}


Our contributions can be summarized as follows:

\begin{enumerate}
\item[--] We show that in general the optimal lightpath routing depends on the link failure probability.
\item[--] We show that for given logical and physical topologies, if there exists a uniformly optimal lightpath routing, then any locally optimal lightpath routing is uniformly optimal.
\item[--] We develop a novel "lexicographical  ordering" for lightpath routing and derive precise optimality conditions in both the low and high failure probability regimes.
\item[--] We develop lightpath rerouting algorithms for maximizing reliability in the low failure probability regime.
\item[--] We develop a logical topology augmentation algorithm for improving the reliability of a given layered network.
\end{enumerate}

The rest of the paper is organized as follows: In Section \ref{sec:model}, we present the network model, and introduce the polynomial expression for the cross-layer reliability and important connectivity parameters related to reliability. In Section \ref{sec:properties}, we study the properties of optimal lightpath routings in the low failure probability regime.  In Section \ref{sec:lightpath-rerouting}, we develop lightpath rerouting and logical topology augmentation algorithms for reliability maximization, and in Section \ref{sec:infocom11_simulation_2}, we present extensive simulation results. In Section \ref{sec:high-regime}, we discuss the optimality conditions for maximum reliability in the high failure probability regime.

\section{Model and Background}\label{sec:model}
We consider a layered network $\mathcal{G}$ that consists of the logical topology $G_L=(V_L, E_L)$ built on top of the physical topology $G_P=(V_P, E_P)$ through a lightpath routing, where $V$ and $E$ are the set of nodes and links respectively. In the context of WDM networks, a logical link is called a \emph{lightpath}, and each lightpath is routed over the physical topology. This \emph{lightpath routing} is denoted by $f=[f_{ij}^{st},(i,j)\in E_P, (s,t)\in E_L]$, where $f_{ij}^{st}$ takes the value 1 if logical link $(s,t)$ is routed over physical link $(i,j)$, and 0 otherwise.

Each physical link fails independently with probability $p$\footnote{\bf Although we assume uniform link failure probability throughout the paper, our results can be readily extended to the case of non-uniform link failure probability by replacing each link with multiple links in series that fail with the same probability. See \cite{lee:reliability} for more details.}. This probabilistic failure model represents a snapshot of a network where links fail and are repaired according to some Markovian process.  Hence, $p$ represents the steady-state probability that a physical link is in a failed state.  This model has been adopted by several previous works \cite{zhang:review, zhang:service, tornatore:availability, segovia:topology}.

If a physical link $(i,j)$ fails, all of the logical links $(s,t)$ carried over $(i,j)$ (i.e., $(s,t)$ such that $f_{ij}^{st}=1$) also fail. A set $S$ of physical links is called a \emph{cross-layer cut} if the failure of the links in $S$ causes the logical network to be disconnected. We also define the \emph{network state} as the subset $S$ of physical links that failed. Hence, if $S$ is a cross-layer cut, the network state $S$ represents a {\it disconnected} network state. Otherwise, it is a {\it connected} state.

\subsection{Failure Polynomial and Connectivity Parameters}\label{sub:failure-polynomial}
Assume that there are $m$ physical links, i.e., $|E_P|=m$. The probability associated with a network state $S$ with exactly $i$ physical link failures (i.e., $|S|=i$) is $p^i(1-p)^{m-i}$. Let $N_i$ be the number of cross-layer cuts $S$ with $|S|=i$, then the probability that the network is disconnected is simply the sum of the probabilities over all cross-layer cuts, i.e.,
\begin{equation}\label{eqn:reliability_polynomial}
F(p) = \sum_{i=0}^{m} N_i p^i(1-p)^{m-i}.
\end{equation}
Therefore, the failure probability of a multi-layer network can be expressed as a polynomial in $p$. The function $F(p)$ will be called the \emph{cross-layer failure polynomial} or simply the \emph{failure polynomial}. The coefficients $N_i$'s contain the information on the structure of a layered graph, determined by the underlying lightpath routing. Below we introduce some important coefficients related to connectivity.

Each $N_i$ represents the number of cross-layer cuts of size $i$ in the network. Define a \emph{Min Cross Layer Cut (MCLC)} as a smallest set of physical links needed to disconnect the logical network. Denote by $d$ the size of an MCLC, then $d$ is the smallest $i$ such that $N_i>0$, meaning that the logical network will not be disconnected by fewer than $d$ physical link failures. The MCLC is a generalization of single-layer min-cut to the multi-layer setting \cite{lee:crosslayer}. It was shown in \cite{lee:reliability} that maximizing the MCLC has the effect of maximizing the reliability in the low failure probability regime.

\subsection{Motivation for Lightpath Rerouting and Logical Topology Augmentation}
\begin{figure}[h]
\centering
\subfigure[Initial routing]{\label{fig:initial-routing}\epsfig{file=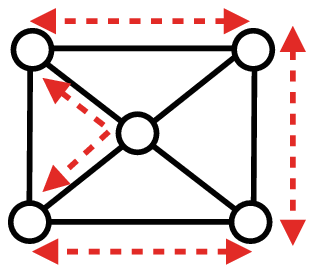,angle=0,width=0.15\textwidth}}\hspace{1.0cm}
\subfigure[After rerouting]{\label{fig:after-rerouting}\epsfig{file=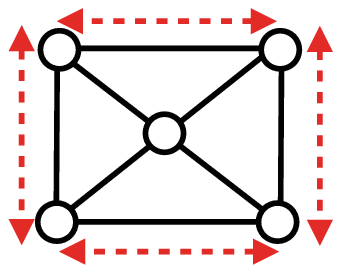,angle=0,width=0.15\textwidth}}
\vspace{-0.2cm}\caption{Example showing that lightpath rerouting can improve the reliability. Physical topology is solid line, logical topology is the rectangle formed by the 4 corner nodes and 4 edges, and lightpath routing is dashed line.}
\label{fig:rerouting-example}\vspace{-0.3cm}
\end{figure}

Although the MCLC criterion is useful for finding a lightpath routing with better reliability, it is not sufficient for fully characterizing reliable lightpath routings. For example, consider the two lightpath routings in Fig \ref{fig:rerouting-example}. The two lightpath routings have the same MCLC value of 2. However, for every value of $p$, the routing in Fig. \ref{fig:after-rerouting} yields better reliability than the one in Fig. \ref{fig:initial-routing}. This example shows that there are more precise conditions for optimal lightpath routings, beyond the MCLC maximization criterion. In Section \ref{sec:properties}, we develop new optimization criteria that characterize in greater detail optimal lightpath routings in the low failure probability regime.

Furthermore, the routing in Fig. \ref{fig:after-rerouting} can be obtained by  rerouting one lightpath from the routing in Fig. \ref{fig:initial-routing}. Hence, this example also demonstrates that one may be able to \emph{find a more reliable lightpath routing by simply rerouting some existing lightpaths} from a given lightpath routing. In Section \ref{sec:lightpath-rerouting}, we study lightpath rerouting algorithms that use the new optimization criteria to find a lightpath routing with better reliability given an initial lightpath routing.

\begin{figure}[h]
\centering
\subfigure[Augmented $G_L$]{\label{fig:augmented-LT}\epsfig{file=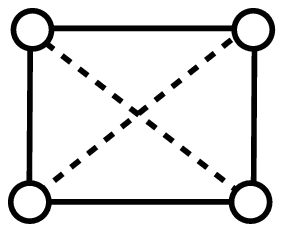,angle=0,width=0.15\textwidth}}\hspace{1.0cm}
\subfigure[Lightpath routing]{\label{fig:augmented-LR}\epsfig{file=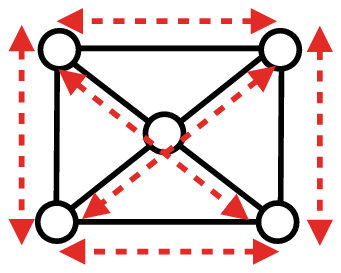,angle=0,width=0.15\textwidth}}
\vspace{-0.2cm}\caption{Example showing that the reliability can be further improved via logical topology augmentation: in (a), dashed lines are added lightpaths.}
\label{high_low_regime}\vspace{-0.1cm}
\end{figure}

In addition to the lightpath rerouting approach, the new optimization criteria can also be used to further enhance the reliability in a different manner. In particular, we consider logical topology augmentation. For instance, suppose that two (diagonal) logical links are added to the logical topology in the example of Fig. \ref{fig:rerouting-example} (see Fig. \ref{fig:augmented-LT}). Fig. \ref{fig:augmented-LR} is an example of routing the two new lightpaths. The new network has far better reliability than the old one in the low failure probability regime since the MCLC value has been raised from 2 to 3. This example shows that augmenting the logical topology can significantly improve the reliability. In Section \ref{sec:topology-augmentation}, using the new optimization criteria, we study how to choose the new logical link that achieves maximum reliability improvement. 


\section{Properties of Optimal Lightpath Routings}\label{sec:properties}
We first study the properties of optimal lightpath routings. These properties will give insight on how routings should be designed for better reliability. Since the failure probability $p$ is typically small in many practical scenarios, we mainly focus on the low failure probability regime. The properties of optimal lightpath routings for large $p$ will be briefly discussed in Section \ref{sec:high-regime}.

\subsection{Uniformly and Locally Optimal Routings}\label{sec:uniform-optimal}
We start with a discussion of routings that are most reliable for all failure probabilities. The observations in this section will motivate a local (in $p$) optimization approach to the design of lightpath routing, which is relatively easy compared with an optimization over all the values of $p$. We begin with the following definition:
\begin{definition}
For given logical and physical topologies, a lightpath routing is said to be \emph{uniformly optimal} if its reliability is greater than or equal to that of any other lightpath routing for every value of $p$.
\end{definition}

Therefore, a uniformly optimal lightpath routing yields the best reliability for all $p\in[0,1]$. Based on the failure polynomial of a lightpath routing, one can immediately develop a sufficient condition for a uniformly optimal lightpath routing:
\begin{observation}
Given a lightpath routing $R$, let $N_i^{R}$ be the number of cross-layer cuts with size $i$. Then $R$ is a uniformly
optimal lightpath routing if, for any other lightpath routing $R^{'}$, $N_i^R \leq N_i^{R^{'}}$ for all $i\in\set{0,\ldots,m}$, where $m$ is the number of physical links.
\label{thm:uniform_optimal_dominance}
\end{observation}

While it is desirable to design a uniformly optimal routing, such a routing does not always exist. Intuitively, for small $p$, only a small number of links are likely to fail simultaneously, and hence for better reliability it is important to remain connected after a small number of failures. In contrast, for large $p$, it is likely that a large number of links fail simultaneously, and thus it is important to withstand a large number of failures. These two objectives conflict because the former prefers disjoint lightpath routing whereas the latter prefers shortest lightpath routing. \begin{figure}[h]
\centering
\subfigure[Optimal Routing in Low Regime]{\label{fig:low_regime}\epsfig{file=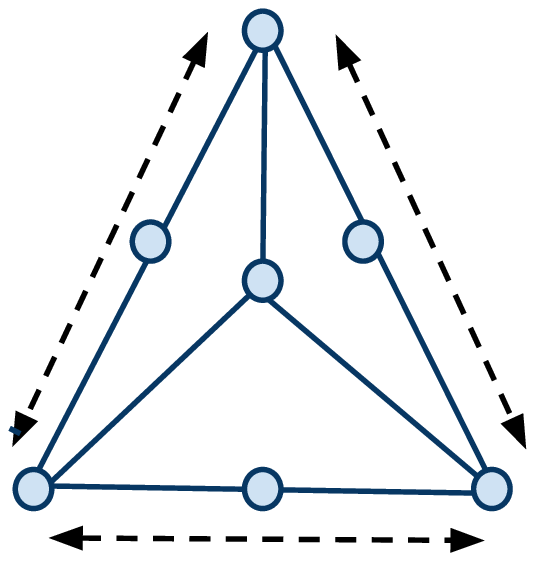,angle=0,width=1in,height=1in}}\hspace{1.0cm}
\subfigure[Optimal Routing in High Regime]{\label{fig:high_regime}\epsfig{file=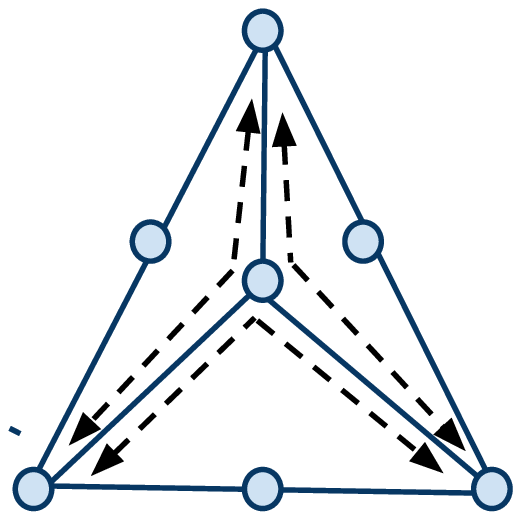,angle=0,width=1in,height=1in}}
\vspace{-0.2cm}\caption{Example showing that optimal routings depend on the value of $p$. Physical topology is solid line, logical topology is the triangle formed by the 3 corner nodes and 3 edges, and lightpath routing is dashed line.}
\label{high_low_regime}\vspace{-0.1cm}
\end{figure}

For example, Fig. \ref{high_low_regime} shows two different lightpath routings. In Fig. \ref{fig:low_regime}, the logical links are routed over physically disjoint paths, and its reliability is given by $3(1-p)^4-2(1-p)^6$. In contrast, in Fig. \ref{fig:high_regime}, every pair of logical links share a physical link, and its reliability is $(1-p)^3$. While disjoint path routing is considered to be more reliable, it is easy to see that in this example the disjoint routing has better reliability only for small values of $p$ whereas for large $p$ (e.g., $p>0.7$) the non-disjoint routing is more reliable.


Since uniformly optimal lightpath routings are not always attainable, we are motivated to focus on {\em locally} optimal routings, where the probability regime of
optimality is restricted to a subrange within $[0,1]$. A locally optimal lightpath routing is defined as follows:

\begin{definition}
For given logical and physical topologies, a lightpath routing is said to be \emph{locally optimal} if there exists $0\leq a<b\leq 1$, such that its reliability is greater than or equal to that of any other lightpath routing for every value of $p\in [a,b]$. In addition, the interval $[a,b]$ is called the {\em optimality regime} for the
lightpath routing.
\end{definition}

Note that a uniformly optimal lightpath routing is also locally optimal with optimality regime $[0,1]$.
Theorem \ref{thm:local=uniform} below is a crucial result to this study; namely, it reveals a connection between local optimality and uniform optimality.
\begin{theorem}\label{thm:local=uniform}
Consider a pair of logical and physical topologies $(G_L,G_P)$ for which there exists a uniformly optimal routing. Then, any locally optimal lightpath routing for $(G_L,G_P)$ is also uniformly optimal.
\end{theorem}
\begin{IEEEproof}
Denote by $F^*(p)$ the failure polynomial of a uniformly optimal lightpath routing. By definition, $F^*(p)$ is no greater than any other failure polynomial for $p\in[0,1]$. Consider a locally optimal lightpath routing $L$ with optimality regime $[p_1,p_2]$, and let $F^L(p)$ be its failure polynomial.

The polynomial equation $F^L(p)-F^*(p)=0$ has degree at most $m$ and thus has at most $m$ roots unless the polynomial $F^L(p)-F^*(p)$ is trivially zero. However, by the definitions of local optimality and uniform optimality, the equation has an infinite number of solutions over the interval $[p_1,p_2]$. Consequently, $F^L(p)$ is identical to $F^*(p)$, which implies that lightpath routing $L$ is also uniformly optimal.
\end{IEEEproof}

Motivated by this result, we study locally optimal lightpath routings. In particular, we develop the conditions for a lightpath routing to be optimal for the low failure probability regime (small $p$).

\subsection{Low Failure Probability Regime}\label{sub:low-regime}
It is easy to see that in the failure polynomial, the terms corresponding to small cross-layer cuts dominate when $p$ is small. Hence, for reliability maximization in the low failure probability regime, it is desirable to minimize the number of small cross-layer cuts. We use this intuition to derive the properties of optimal routings for small $p$. We begin with the following definition:
\begin{definition}
Consider two lightpath routings $1$ and $2$. Routing $1$ is said to be \emph{more reliable} than routing $2$ \emph{in the low failure probability regime} if there exists a positive number $p_0$ such that the reliability of routing $1$ is higher than that of routing $2$ for $0<p<p_0$. A lightpath routing is said to be \emph{locally optimal in the low failure probability regime} if it is more (or equally) reliable than any other routing in the low failure probability regime.
\end{definition}

In the following, we characterize the impact of small cuts on the reliability. Let $d_j$ be the size of the MCLC under routing $j(=1,2)$. Let $N_i$ and $M_i$ be the numbers of cross-layer cuts of size $i$ under routings $1$ and $2$ respectively. We call the vector $N=[N_i,\forall i]$ the \emph{cut vector}. The following is an example of cut vectors $N$ and $M$ with $d_1=4$ and $d_2=3$:




\[
\begin{array}{rllllllll}
i & 0 & 1 & 2 & 3 & 4 & 5 & \cdots & m\\
N_i & 0 & 0 & 0 & 0 & 20 & 26 & \cdots & 1\\
M_i & 0 & 0 & 0 & 9 & 19 & 30 & \cdots & 1.
\end{array}
\]
Using cut vectors of lightpath routings, we define \emph{lexicographical ordering} as follows:
\begin{definition}\label{def:lexico}
Routing $1$ is lexicographically smaller than routing $2$ if $N_{d}<M_{d}$ where $d$ is the smallest $i$ at which $N_i$ and $M_i$ differ.
\label{def:lexico}
\end{definition}
Note that a lightpath routing with a larger MCLC size is lexicographically smaller by Definition \ref{def:lexico}. In the above example, we have $d=3$ and $N_{d}<M_{d}$, hence routing $1$ is lexicographically smaller. Therefore, if a lightpath routing is lexicographically smaller than another, it has fewer small cross-layer cuts and thus yields better reliability for small $p$.
\begin{theorem}\label{thm:lexico-ordering}
{\kyb
Given two lightpath routings 1 and 2 with cut vectors $[N_i|i=0,\ldots,m]$ and
$[M_i|i=0,\ldots,m]$ respectively, where $m$ is the number of physical links,
if routing $1$ is lexicographically smaller than routing $2$, then routing $1$ is more reliable than routing $2$ in the low failure probability regime. In particular, let $d=\min\limits_{0\leq i\leq m}\{i:M_i \neq N_i\}$ be the index where the elements in the cut vectors first differ. Then, lightpath routing 1 is more reliable than routing 2 for $p<p_0=\frac{(d+1)(M_d-N_d)}{2m\choose{m}{d}}$.
}
\end{theorem}
\begin{proof}
This is implied by~\thmref{partial_sum_dominance}, which will be discussed in~\secref{infocom11_regime_extension}.
\end{proof}

Clearly, Theorem~\ref{thm:lexico-ordering} leads to a local optimality condition; that is, if a lightpath routing minimizes the cut vector lexicographically, then it is locally optimal in the low failure probability regime. An interesting case is when routing $1$ has larger MCLC than routing $2$ (as in the above example). In this case, routing $1$ is lexicographically smaller than routing $2$, and Theorem~\ref{thm:lexico-ordering} implies the following corollary.

\begin{corollary}\label{cor:low-regime1}
If $d_1>d_2$, then routing $1$ is more reliable than routing $2$ in the low failure probability regime.
\end{corollary}

Consequently, a lightpath routing with the maximum size MCLC yields the best reliability for small $p$. We note that the same result was shown in \cite{lee:reliability}.
Similarly, routing $1$ is also lexicographically smaller than routing $2$ when they have the same size of MCLC but routing $1$ has fewer MCLCs. This leads to the following result:
\begin{corollary}\label{cor:low-regime2}
If $d_1=d_2$ and $N_{d_1}<M_{d_2}$, then routing $1$ is more reliable than routing $2$ in the low probability regime.
\end{corollary}

The expression for $p_0$ given in~\thmref{lexico-ordering} also provides some insight into how the difference of the cut vectors affects the guaranteed regime. For example, if $d$ is small and $M_d-N_d$ is large, the guaranteed regime is larger. In other words, if one lightpath routing has fewer small cross-layer cuts than the other, it will achieve higher reliability for a larger range of $p$ in the low probability regime.

Therefore, for reliability maximization in the low failure probability regime, it is desirable to maximize the size of the MCLC while minimizing the number of such MCLCs. This condition will be used to develop the algorithms in Section \ref{sec:lightpath-rerouting}.

\subsection{Extension of Optimal Probability Regimes}\label{sec:infocom11_regime_extension}
The expressions in~\thmref{lexico-ordering} only consider the first element in the two cut vectors that are different. As a result, the guaranteed regime is rather conservative. For instance, the expression fails to capture the uniform optimality for a lightpath routing that  satisfies the condition in~\thmref{uniform_optimal_dominance}. In this section, we will develop a more general expression for the regime bound that includes other elements in the cut vectors.

Consider two lightpath routings $1$ and $2$. Let $F_j(p)$ be the failure polynomial of routing $j$ ($=1,2$), and $N_i$'s and $M_i$'s be the coefficients in $F_1(p)$ and $F_2(p)$ respectively. Define the following vector of partial sums:
\[
\overrightarrow{{\bf N}}=\left[\sum\limits_{i=0}^kN_i|k=0,...,m\right]
\]
The vector $\overrightarrow{{\bf M}}$ is defined similarly. Note that the $i$-th element $\overrightarrow{N}_i$ of vector $\overrightarrow{{\bf N}}$ is the total number of cross-layer cuts of size at most $i$. 
We first extend the definition of \emph{lexicographical ordering} as follows:
\begin{definition}
Lightpath routing $1$ is said to be  $k$-lexicographically smaller than lightpath routing $2$ if
\begin{align*}
k=\max\set{j:\overrightarrow{N}_{i}\leq\overrightarrow{M}_{i},\quad\forall i< d+j} \mbox{\ and $k\geq1$},
\end{align*}
where $d$ is the position of the first element where the two cut vectors differ.
\end{definition}

Therefore, a lightpath routing is lexicographically smaller (in the original sense) if and only if it is $k$-lexicographically smaller for some $k\geq 1$.
The $k$-lexicographical ordering thus compares two lightpath routings based on structures beyond the smallest cuts, making it possible to establish a larger optimality regime. Roughly speaking, the value of $k$ reflects the degree of
dominance of a lightpath routing in the low probability regime: a $k$-lexicographically smaller lightpath routing
means that it has fewer ``small'' cuts, where the definition for ``small'' is broader
if $k$ is larger.

It is obvious that when $p\leq0.5$, the failure probability of a cross-layer cut is a non-increasing function of the cut size, because $p^{i}(1-p)^{m-i}\geq p^{i+1}(1-p)^{m-(i+1)}$ for $p\leq0.5$. Suppose that routing $1$ has smaller total number of cuts of size up to $i$ than routing $2$, i.e., $\overrightarrow{N}_i\leq\overrightarrow{M}_i$. To compare cross-layer cuts of size at most $i+1$, suppose further that the relative increment $N_{i+1}-M_{i+1}$ in the number of larger cuts does not exceed the surplus $\overrightarrow{M}_i-\overrightarrow{N}_i$ from smaller cuts, i.e., $\overrightarrow{N}_{i+1}\leq\overrightarrow{M}_{i+1}$. Then, with respect to cut size at most $i+1$, routing $1$ will have smaller failure probability than routing $2$, provided that the same was true for cut size up to $i$. This observation leads to the following theorem on the relationship between lexicographical ordering and probability regime.

\begin{theorem}\label{thm:partial_sum_dominance}
Given two vectors {\bf N}=$[N_i|i=0,\ldots,m]$ and
{\bf M}=$[M_i|i=0,\ldots,m]$, let $F_1(p)=\sum_{i=0}^nN_ip^i(1-p)^{m-i}$ and $F_2(p)=\sum_{i=0}^nM_ip^i(1-p)^{m-i}$. For any $j$,
let $\overrightarrow{\Delta}_j=\sum\limits_{i=0}^j(M_i-N_i)$ and $\overrightarrow{\delta}_j = \max\limits_{j+1\leq i\leq m}\set{\frac{N_i-M_i}{\choose{m}{i}}}$.
If the vector {\bf N} is $k$-lexicographically smaller than {\bf M}, then
\[
F_1(p)\leq F_2(p) \textrm{ for } p\leq p_0^l=\min\set{0.5, \max\limits_{d\leq j\leq d+k-1}B_j}
\]
where $d=\min\set{i:N_i< M_i}$ and
\begin{align*}
B_j&=\left\{\begin{array}{l l}
0.5,&\mbox{if $j=m$} \nonumber \\
\frac{1}{\frac{m}{j+1} + \overrightarrow{\delta}_j\choose{m}{j+1}/\overrightarrow{\Delta}_j},& \mbox{otherwise.} \nonumber \\
\end{array}
\right.
\end{align*}

\begin{proof}
See Appendix \ref{app:unifying-thm}.
\end{proof}
\end{theorem}

Therefore, the probability regime bound $p_0^l$ in Theorem \ref{thm:partial_sum_dominance} is a non-decreasing function of $k$, which means that a lightpath routing with smaller number of cuts over a larger size range will
be more reliable over a larger probability regime. This is consistent with the conclusion in Section \ref{sub:low-regime}, that the lightpath routing design should
minimize the lexicographical ordering of the cut vector.

{\kyb
~\thmref{lexico-ordering} follows from~\thmref{partial_sum_dominance}.
 For a lexicographically smaller lightpath routing, the term $B_d$
 is given by:
\begin{align*}
\frac{1}{\frac{m}{d+1} + \overrightarrow{\delta}_d\choose{m}{d+1}/\overrightarrow{\Delta}_d }
&=\frac{1}{\frac{m}{d+1} + \overrightarrow{\delta}_d\choose{m}{d+1}/(M_d-N_d)} \\
&\geq\frac{(d+1)(M_d-N_d)}{m(M_d-N_d) + (d+1)\choose{m}{d+1}},\\
&\geq\frac{(d+1)(M_d-N_d)}{m\choose{m}{d} + (m-d)\choose{m}{d}} \\
&\geq\frac{(d+1)(M_d-N_d)}{2m\choose{m}{d}},
\end{align*}
where the first inequality is due to $\overrightarrow{\delta}_d\leq 1$.

An interesting special case is when $d+k-1=m$, that is, $\overrightarrow{M}_j\geq\overrightarrow{N}_j$ for all $j=0,\ldots,m$.
In that case, the term $B_{d+k-1}=B_m=0.5$, implying that the optimality regime is $[0,0.5]$. We summarize this as the following corollary:
\begin{corollary}\label{cor:p-leq-.5}
If $\overrightarrow{N}_j\leq\overrightarrow{M}_j$ for all $j=0,\ldots,m$,
then lightpath routing $1$ is at least as reliable as lightpath routing $2$ for $p\leq0.5$, i.e., $F_1(p)\leq F_2(p)$ for $p\leq0.5$.
\end{corollary}

Note that the condition in~\corref{p-leq-.5} requires every partial sum in the vector {\bf M} to be at least the corresponding partial sum in the vector {\bf N}, which is a stronger condition than the lexicographic comparison in~\thmref{lexico-ordering}.
This stronger condition allows the better optimality regime to be established
in~\corref{p-leq-.5}.
}

\section{Maximizing Reliability by Improving Lightpath Routing and Logical Connectivity}
\label{sec:lightpath-rerouting}
In this section, we explore ways to improve the reliability
of a layered network. Typically, the physical topology is static and difficult to change.
Therefore, the reliability of a layered network can be improved by one of two ways: (i) improving the lightpath routing, or (ii) improving the logical topology.


We have shown in~\secref{properties} that when physical link failures are rare,
the lightpath routing that minimizes the lexicographical ordering
will maximize the reliability. This new observation gives us an 
exact optimization criterion for designing reliable layered networks. 

{\bf As discussed in~\secref{intro}, the traditional approach of
  jointly routing all logical links is often too complex, which makes it
infeasible for larger networks.
This motivates  the incremental approach introduced in this section, 
where the layered network is improved one
logical link at a time. This significantly reduces the problem space
and allows us to use a more sophisticated objective function based on
the optimziation criterion we studied in~\secref{properties}.

Within this context, we study two optimization problems that are 
fundamental to improving the lightpath routing and logical
connectivity:
}

\begin{enumerate}
\item {\bf Lightpath Rerouting}: Given the physical, logical topologies and a lightpath routing, find a logical link to reroute, such that the resulting reliability is maximized.
\item {\bf Logical Topology Augmentation}: Given the physical, logical topologies and a lightpath routing, find a pair of logical nodes, as well as a physical path between the nodes, such that the addition of the corresponding logical link will provide maximum reliability improvement.
\end{enumerate}

The above two problems are basic building blocks for designing reliable layered networks. For example, given an existing layered network, we can iteratively reroute existing lightpaths in the network until no further improvement is possible (e.g.~\figref{rerouting_improve}). {\bf Hence, given the physical and logical topologies, the iterative rerouting algorithm can be described as follows: 
\begin{enumerate}
\item Generate an arbitrary initial lightpath routing.
\item Reroute a logical link using ILP/approximation algorithm introduced in~\secref{lightpath-rerouting}.
\item Repeat Step 2 until no further improvement can be made by rerouting a single lightpath.
\end{enumerate}
Similarly, if it is feasible to add new logical links, we can iteratively augment the logical topology to further improve the reliability; and studying the Logical Topology Augmentation problem allows us to select such new logical links effectively. These iterative rerouting and augmentation algorithms will be used for performance evaluation in Section \ref{sec:infocom11_simulation_2}.}

\begin{figure}[h]
\centering
\subfigure[$d=1, N_d=3$]{
          \includegraphics[height=.11\textwidth]{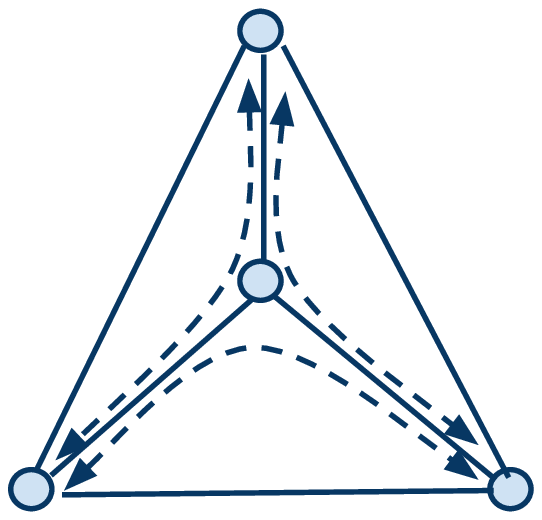}}
\subfigure[$d=1, N_d=1$]{
          \includegraphics[height=.11\textwidth]{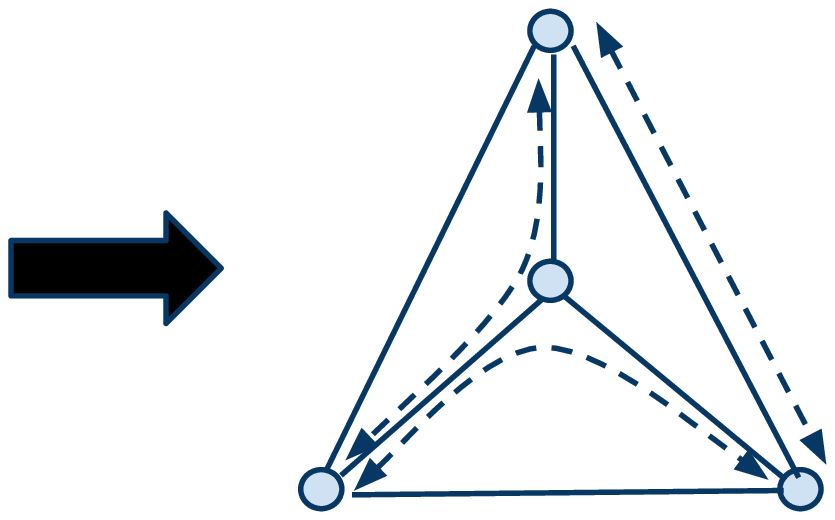}} \\
\subfigure[$d=2, N_d=5$]{
          \includegraphics[height=.11\textwidth]{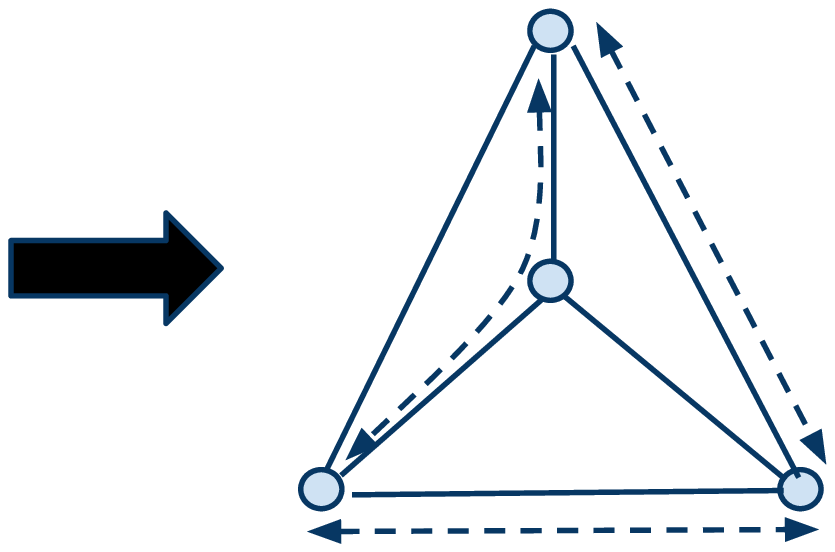}}
\subfigure[$d=2, N_d=3$]{
          \includegraphics[height=.11\textwidth]{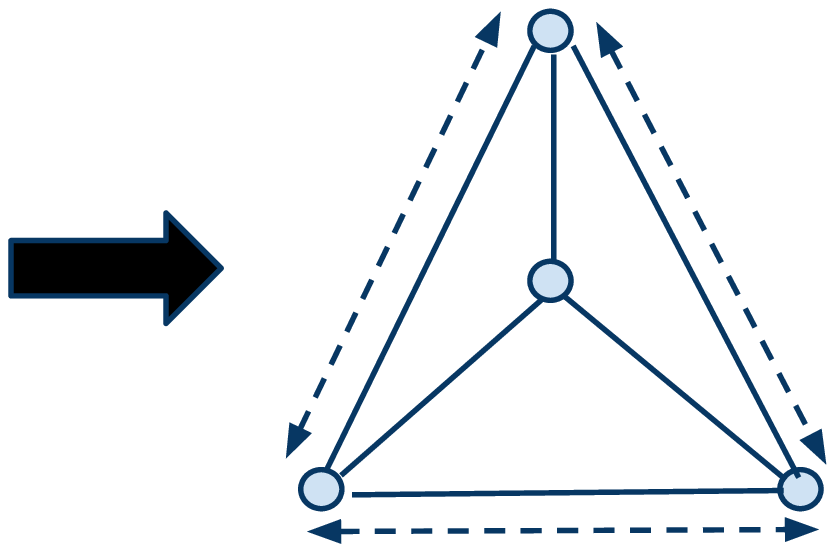}}
\caption{Improving reliability via lightpath rerouting. The physical topology is in solid lines, and the lightpath routing of the logical topology
is in dashed lines. The MCLC value and the number of MCLCs in the lightpath routings are denoted by $d$ and $N_d$.}
\label{fig:rerouting_improve}\vspace{-0.4cm}
\end{figure}

In this section, we present algorithms for the rerouting and augmentation problems. In the next section, we 
will evaluate the effectiveness of rerouting and augmentation on improving
cross-layer reliability.

\subsection{Lightpath Rerouting}
\label{sec:rerouting_alg}
Given a layered network and its lightpath routing, the objective of the {\it Lightpath Rerouting} problem is to find 
the best way to reroute a lightpath, so that the reliability improvement is maximized. Recall that with low link failure probability, the reliability of a network is maximized when the lexicographical ordering of its cut vector is
minimized. Therefore, the most effective reroute should maximize the MCLC of the
resulting lightpath routing, and also minimize the number of MCLCs.

In the following sections, we first analyze the effect of rerouting a lightpath and characterize conditions where
such a rerouting is beneficial. This provides the groundwork for our rerouting algorithms. Based on these observations, we develop
an ILP to find the optimal lightpath to reroute.
Next, we propose an approximation algorithm that computes a near-optimal
solution in much shorter time. This gives us a scalable algorithm that can be
used for designing large layered networks.

\subsubsection{Effects of Rerouting a Lightpath}
\label{sec:rerouting-background}
Let $d$ be the size of the MCLC under the initial routing. When the physical route of a logical link changes, some cross-layer cuts will be converted into non-cuts, and some non-cuts will be converted into cross-layer cuts. In the low failure probability regime, the reliability will be improved by the rerouting if the following is true:

\begin{enumerate}
\item The  conversion of cross-layer cuts with size $d$ to non-cuts outnumbers the conversion in the opposite direction.
\item The MCLC value does not decrease.
\end{enumerate}

 Therefore, we can formulate the lightpath rerouting as an optimization problem to maximize the
 reduction in the number of MCLCs, subject to the constraint that no non-cuts of size smaller than
 $d$ is converted to cross-layer cuts. The exact conditions for the conversion
between cuts and non-cuts are described as follows, which will be
used as the basis of the ILP formulation as well as the approximation algorithm.

Given the physical topology $G_P=(V_P, E_P)$ and the logical topology $G_L=(V_L,E_L)$, we
 model a lightpath routing as a set of binary constants $\set{f_{ij}^{st}}$, where $f_{ij}^{st}=1$
if and only if logical link $(s,t)$ uses physical link $(i,j)$ in the lightpath routing. For a
given  set of physical links $S$, we
define the {\it logical residual graph} for $S$, denoted as  $G_L^S$, to be
$\set{(s,t)\in E_L:\sum\limits_{ (i,j)\in S}f^{st}_{ij}=0}$. In other words, the
residual graph consists of logical links that use none of the physical links in $S$.
By definition, the set $S$ is a cross-layer cut if and only
if its logical residual graph is disconnected. Given a cross-layer cut $S$, it is called a
{\it $k$-way cross-layer cut} if its logical residual graph has $k$ connected components.
In addition, given a cross-layer non-cut $T$ for a lightpath routing, we call a logical
link $(s,t)$ {\it critical} to $T$ if $(s,t)$ is a cut edge of the residual graph $G_L^T$, that is,
it is an edge in $G_L^T$ whose removal will disconnect the residual graph.

The following theorems describe the conditions for a lightpath rerouting that results in conversions between cross-layer cuts and non-cuts. The proofs can be found in \cite{kayi:thesis}.

\begin{theorem}
\label{thm:cut_to_non_cut}
Let $S$ be a cross-layer cut for a lightpath routing. Rerouting logical link $(s,t)$ from physical
path $P_1$ to $P_2$ turns $S$ into a non-cut if and only if the following conditions are true:
\begin{enumerate}
\item $S$ is a 2-way cross-layer cut.
\item $s$ and $t$ are disconnected in the residual graph for $S$.
\item $P_2$ does not use any physical links in $S$.
\end{enumerate}

\begin{proof}
Suppose all the above conditions are true. Since the new route $P_2$ does not use any physical links in $S$, the
logical link $(s,t)$ will be in the logical residual graph for $S$ under the new lightpath routing. Other logical
links that are in the original residual graph will remain,
because none of their physical routes have changed. Therefore, the
residual graph will become connected now that $(s,t)$ is added to it, which
implies $S$ becomes a non-cut. It can be easily verified that the residual graph will
remain disconnected if any of the above conditions do not hold.
\end{proof}
\end{theorem}

\begin{theorem}
\label{thm:non_cut_to_cut}
Let $T$ be a cross-layer non-cut for a lightpath routing. Rerouting logical link $(s,t)$ from physical
path $P_1$ to $P_2$ turns $T$ into a cross-layer cut if and only if the following conditions are true:
\begin{enumerate}
\item $(s,t)$ is critical to $T$.
\item $P_2$ uses some physical link in $T$.
\end{enumerate}

\begin{proof}
Suppose both conditions are true.
Since $P_2$ uses some physical fiber in $T$, the logical link will be removed
from the residual graph for $T$ under the new lightpath routing.
Since $(s,t)$ is critical to the non-cut $T$, its removal will disconnect the residual graph, which means
that $T$ will become a cross-layer cut. It can be easily verified that
the residual graph will remain connected if any of the two conditions do not hold.
\end{proof}
\end{theorem}

Therefore, the optimal rerouting should maximize the number of cross-layer cuts
satisfying~\thmref{cut_to_non_cut} and minimize the number of non-cuts
satisfying~\thmref{non_cut_to_cut}.
However, it is also important to ensure that none of the non-cuts with size
smaller than $d$ is converted to cross-layer cuts by the rerouting, since otherwise the MCLC
value will decrease. The following theorem states that only non-cuts with size at least $d-1$
can be converted into a cross-layer cut by rerouting a single lightpath.

\begin{theorem}
\label{thm:not_cut_size}
Let $d$ be the Min Cross Layer Cut value of a lightpath routing and let $\mathcal{NC}$ be the set of
cross-layer non-cuts that can be converted into cross-layer cuts by rerouting a single logical link.
 Then $|T|\geq d-1$ for all $T\in\mathcal{NC}$.
 
 \begin{proof}
Suppose $\mathcal{NC}$ contains a convertible non-cut $S$ with size less than $d-1$. Since $S$ is convertible,
there exists a logical link $(s,t)$ that is critical to $S$. Now let $l$ be a fiber used by $(s,t)$, then
the fiber set $S\cup\set{l}$ would disconnect the logical residual graph and is therefore a cross-layer cut.
However, such a set contains at most $d-1$ fibers, contradicting that $d$ is the Min Cross Layer Cut.
\end{proof}
\end{theorem}

Therefore, when rerouting a lightpath, we need to make sure
that none of the non-cuts with size $d-1$ get converted into cuts
in order to prevent the MCLC value from decreasing.
Based on these observations, we next develop an ILP for the lightpath rerouting problem.

\subsubsection{ILP for Lightpath Rerouting}
\label{sec:rerouting_ilp}
For the given lightpath routing, let $d$ be
the MCLC value, and let $\mathcal{C}_d, \mathcal{NC}_d$ and
$\mathcal{NC}_{d-1}$ be the sets of
2-way cross-layer cuts with size $d$, non-cuts with size $d$,
and non-cuts with size $d-1$ respectively.
The lightpath rerouting problem can be formulated as an ILP that finds the logical link,
 and its new physical route, that maximizes the net reduction in
 MCLCs. {\bf In other words, the optimal reroute should result in the
  minimum number of  cross-layer cuts with size $d$, without creating any cross-layers
  cuts with size $d-1$.}

The ILP can be considered as a path selection problem on an auxiliary graph $G_P^{'}=(V^{'}_P, E_P^{'})$,
where $V^{'}_P=V_P\cup\set{u,v}$, with $u$ and $v$ being the additional source and sink nodes in the
auxiliary graph; and $E_P^{'}=E_P\cup\set{(u,x),(x,v):x\in V_P}$.~\figref{rerouting_ilp} illustrates
the construction of the auxiliary graph.

\begin{figure}[h]
\centering
\includegraphics[height=0.7in]{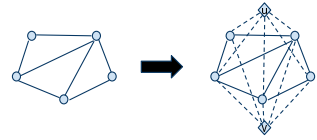}\vspace{-0.2cm}
\caption{Construction of the auxiliary graph for the ILP. $u$ and $v$ are the additional source
and sink nodes, and the dashed lines are the additional links in the auxiliary graph.}
\label{fig:rerouting_ilp}\vspace{-0.3cm}
\end{figure}

We first define the following variables and parameters:

\begin{enumerate}
\item Variables:
\begin{itemize}
\item $\set{g_{st}: (s,t)\in E_L}$: 1 if logical link $(s,t)$ is rerouted, and 0 otherwise.
\item $\set{f_{ij}: (i,j)\in E^{'}_P}$: Flow variables describing a path in $G_P^{'}$ from node $u$ to node $v$.
\item $\set{y^c : c\in\mathcal{C}_d}$: 1 if the cross-layer cut $c$ is converted into a non-cut by the lightpath rerouting, and 0 otherwise.
\item $\set{z^c :c\in\mathcal{NC}_d}$: 1 if the non-cut $c$ is converted into a cross-layer cut by the lightpath rerouting, and 0 otherwise.
\end{itemize}
\item Parameters:
\begin{itemize}
\item $\set{h^c_{st} : c\in\mathcal{C}_d, (s,t)\in E_L}$: 1 if logical nodes $s$ and $t$ are disconnected by the 2-way cut $c$, and 0 otherwise.
\item $\set{q^c_{st} : c\in\mathcal{NC}_{d}\cup\mathcal{NC}_{d-1}, (s,t)\in E_L}$: 1 if logical link $(s,t)$ is critical to the non-cut $c$, and 0 otherwise.
\item $\set{l^c_{ij} : \forall c\in\mathcal{C}_d\cup\mathcal{NC}_d\cup\mathcal{NC}_{d-1}, (i,j)\in E_P}$: 1 if physical link $(i,j)$ is in set $c$, and 0 otherwise.
\end{itemize}
\end{enumerate}
The lightpath rerouting can be formulated as follows:
\begin{eqnarray}
\mathsf{REROUTE:}\quad\mbox{Maximize\ } \sum_{c\in\mathcal{C}_d}y^c-\sum_{c\in\mathcal{NC}_d}z^c, \quad\mbox{subject to:}  \nonumber \\
\label{con:lp_choice}
g_{st}\leq (f_{us} + f_{tv}) / 2, \quad\forall (s,t)\in  E_L\\
\label{con:one_lp}
\sum_{(s,t)\in E_L}g_{st}=1 \\
\label{con:non_clc_forbidden}
l^c_{ij}f_{ij} + \sum_{(s,t)\in  E_L}q^c_{st}g_{st} \leq 1,\quad \forall c \in\mathcal{NC}_{d-1}, (i,j)\in E_P \\
\label{con:non_clc_count}
l^c_{ij}f_{ij} + \sum_{(s,t)\in  E_L}q^c_{st}g_{st}\leq z^c+1,\forall c \in\mathcal{NC}_d, (i,j)\in E_P \\
\label{con:clc_connect}
y^c\leq \sum_{(s,t)\in E_L}h_{st}^c g_{st}, \quad \forall c \in\mathcal{C}_d \\
\label{con:clc_avoid}
y^c\leq 1-l_{ij}^cf_{ij},\quad\forall (i,j)\in E_P, \forall c \in\mathcal{C}_d \\
\label{con:reroute_one_link}
\{(i,j):f_{ij}=1\}\mbox{\ forms an $(u,v)$-path in $G_P^{'}$\ }  \\
f_{ij}, g_{st}\in\set{0,1}, 0\leq y^{c}, z^{c}\leq1 \nonumber
\end{eqnarray}

{\kyb
The formulation can be interpreted as a path selection problem on the auxiliary graph $G_P^{'}$.
\conref{reroute_one_link}, which requires that the variables $f_{ij}$ describe a path
from $u$ to $v$, can be expressed by the standard flow conservation constraints.
As a result, in a feasible solution to the formulation, the variables $f_{ij}$ represent a path
$u\rightarrow s\leadsto t\rightarrow v$, which corresponds to the new physical route for the logical
 link $(s,t)$ after the rerouting.
}

\conref{lp_choice} ensures that $g_{st}$ can be set to 1 only if
$f_{ij}$ represents the path $u\rightarrow s\leadsto t\rightarrow v$, and~\conref{one_lp}
makes sure that the chosen $(s,t)$ is indeed a logical link in $E_L$.
Therefore, exactly one logical link $(s,t)$ can have $g_{st}=1$, and
a feasible solution to this ILP corresponds to a rerouting of
the logical link.

In~\conref{non_clc_forbidden}, the two terms correspond to the conditions in~\thmref{not_cut_size}. The constraint makes sure that at most one of the conditions is satisfied, thereby disallowing the non-cuts of size $d-1$ to be converted into a cross-layer cut.
Similarly,~\conref{non_clc_count} makes sure $z^c=1$ for any non-cut $c\in\mathcal{NC}_d$ that is converted into a cut by the rerouting.

Finally,~\conreftwo{clc_connect}{clc_avoid}
describe conditions 2) and 3) of~\thmref{cut_to_non_cut} respectively.
Therefore, $y^c$ can be 1 only if both conditions are satisfied.
Since $c$ also satisfies condition 1) by definition of $\mathcal{C}_d$,
this implies that cross-layer cut $c$ is converted into a non-cut when $y^c=1$.

Since the objective is to maximize $y^c$ and minimize $z^c$, in an optimal solution
$y^c=1$ if and only if cross-layer cut $c$ is converted into a non-cut, and
$z^c=1$ if and only if non-cut $c$ is converted into a cross-layer cut.
As a result, the objective function reflects the net reduction in the number of MCLCs.

Finally, note that the variables $y^c$ and $z^c$ will take on
binary values in an optimal solution even if they are not constrained to be integral.
 This observation helps to reduce the number of binary variables in the formulation.

{\bf
 The ILP $\mathsf{REROUTE}$ approximates the lexicographical ordering minimization
by minimizing the number of MCLCs in the network. It can be
extended to consider cross-layer cuts of size larger than $d$, thus
achieving a better approximation. In this case, the set of cross layer cuts and non-cuts 
$\mathcal{C}_d$ and $\mathcal{NC}_d$ will be replaced by sets that
include the cut and non-cuts up to size $k > d$, denoted as 
$\mathcal{C}_{\leq k}$ and $\mathcal{NC}_{\leq k}$. The objective
function will be changed to 
\begin{align} 
Maximize\quad \sum_{c\in\mathcal{C}_{\leq k}}y^cw^c-\sum_{c\in\mathcal{NC}_{\leq k}}z^cw^c, 
\end{align}
where $w^c$ is a weight constant assigned to each cut $c$ so that
a smaller cut will have weight that dominates cuts of larger size. 
In particular, if $k$ is set to $|E_P|$, the extended ILP will return the
optimal solution that minimizes the lexicographical ordering.
However, such a formulation will contain an exponential number of
variables $y^c$ and $z^c$, and is generally not feasible for
practical use. Therefore, in the rest of the paper, we will focus on
the problem of minimizing the number of MCLCs, though the techniques
discussed in this paper are also be applicable to the more general setting.}

\subsubsection{Approximation Algorithm for Lightpath Rerouting}
\label{sec:algorithm-approx}
For larger networks, however, solving the rerouting ILP may still be infeasible. Therefore, 
in this section, we present an approximation algorithm for the rerouting problem that provides near-optimal solutions within a much shorter time.

 We focus on the following question: Given the lightpath routing, and a logical link $(s,t)$, what
is the best way to reroute $(s,t)$ assuming the routes for all other logical links are fixed? A solution
to this problem will allow us to solve the lightpath rerouting problem, since we can run the
algorithm once for each logical link, and return the best solution.

Similar to the previous section, let $\mathcal{C}_d, \mathcal{NC}_d$ and $\mathcal{NC}_{d-1}$ be the set
of cross-layer cuts of size $d$, non-cuts of size $d$ and non-cuts of size $d-1$ respectively.
Now suppose $Q$ is a new physical route for logical link $(s,t)$. Let $\mathcal{NC}_d^{st}$ and $\mathcal{NC}_{d-1}^{st}$ be the subsets of $\mathcal{NC}_d$ and
$\mathcal{NC}_{d-1}$ that satisfy condition (1) of~\thmref{non_cut_to_cut}. These two sets represent the non-cuts that can
potentially be converted into a cut by rerouting $(s,t)$. It immediately follows that any $(s,t)$
path that uses a physical link in $\cup_{T\in\mathcal{NC}_{d-1}^{st}}T$ will create a cross-layer
cut with size $d-1$, which should be forbidden for the new physical route. In addition,
for any physical link $(i,j)$, the set
$\mathcal{L}^{\mathcal{NC}}_{ij} = \set{T\in\mathcal{NC}_d^{st}:(i,j)\in T}$ represents the
non-cuts with size $d$ that will be converted into cross-layer cuts if the new route $Q$ contains the physical link $(i,j)$.

Similarly, let $\mathcal{C}_d^{st}\subseteq\mathcal{C}_d$ be the set of cross-layer cuts that do not satisfy conditions
(1) or (2) of~\thmref{cut_to_non_cut}. This represents the set that will continue to be cross-layer cuts
regardless of the new physical route $Q$ for $(s,t)$. In addition, for each $(i,j)\in E_P$, the
cross-layer cuts in the set
$\mathcal{L}^\mathcal{C}_{ij}=\set{S\in\mathcal{C}_d-\mathcal{C}_d^{st}:(i,j)\in S}$
will also continue to be cross-layer cuts if the new route $Q$ contains the physical link $(i,j)$, as they do not satisfy condition 3) of~\thmref{cut_to_non_cut}.

Now, for each physical link $(i,j)$, let $\mathcal{L}_{ij}=\mathcal{L}^{\mathcal{C}}_{ij} \cup \mathcal{L}^{\mathcal{NC}}_{ij}$.
If a physical link $(i,j)$ is used by the logical link $(s,t)$ in the new route $Q$, it will cause the set $\mathcal{L}_{ij}\cup\mathcal{C}_d^{st}$
to become
cross-layer cuts. Since every set of physical links in $\mathcal{C}_d^{st}$ will be cross-layer cuts regardless of the physical route taken by $(s,t)$, the lightpath rerouting problem for logical link $(s,t)$ can be
formulated as choosing the $(s,t)$-path $Q$
in $G^{'}_P=(V_P, E_P-\cup_{T\in\mathcal{NC}_{d-1}^{st}}T)$
that minimizes $|\mathcal{L}(Q)|=|\cup_{(i,j)\in Q}\mathcal{L}_{ij}|$.
Although this is an instance of the NP-Hard
{\it Minimum Color Path}~\cite{jasonjue:min_color_path_infocom} problem, a simple
$d$-approximation algorithm exists, as described below:

\begin{algorithm}
\caption{$\mathsf{REROUTE\_SP}(s,t)$}
\begin{algorithmic}[1]
\STATE Construct a weighted graph on $G^{'}_P=(V_P, E_P-\cup_{T\in\mathcal{NC}_{d-1}^{st}}T)$,
where each edge $(i,j)$ is assigned with weight $w(i,j)=|\mathcal{L}_{ij}|$.
\STATE Run Dijkstra's algorithm to find the shortest $(s,t)$-path in the weighted graph.
\label{alg:reroute_approx}
\end{algorithmic}
\end{algorithm}
\begin{theorem}
\label{thm:reroute_d_approx}
Let $Q^{*}$ be the optimal physical route for $(s,t)$ that results in the minimum number of MCLCs, and let $Q^{\mathsf{SP}}$ be the new route
for $(s,t)$ returned by $\mathsf{REROUTE\_SP}$.
For any $(s,t)$-path $Q$, let $N_d(Q)$ be the number of cross-layer cuts with size $d$ after
rerouting $(s,t)$ with $Q$, where $d$ is the size of the MCLC. Then
$N_d(Q^{\mathsf{SP}})\leq d\cdot N_d(Q^{*})$.
\end{theorem}
\begin{IEEEproof}
See Appendix \ref{app:rerouting-approx}.
\end{IEEEproof}

Therefore, the number of cross-layer cuts of size $d$ given by $\mathsf{REROUTE\_SP}$ is at most $d$ times the optimal reroute. Note that if the optimal new route for $(s,t)$ eliminates every MCLC of size $d$, i.e., $N_d(Q^{*})=0$, the approximation algorithm will find a new route that achieves that as well. We state this observation as
the following corollary.

\begin{corollary}
$\mathsf{REROUTE\_SP}(s,t)$ will return a new route for $(s,t)$ that increases the size of MCLC of the layered network, if such a new route exists.
\end{corollary}

We can extend algorithm $\mathsf{REROUTE\_SP}$, which is based on the Dijkstra's shortest path algorithm, by using the $k$-shortest path
algorithm~\cite{yen:k_shortest} to successively
compute the next shortest path in $G^{'}_P$ and keep track of the path $Q$ with the minimum value
of $|\mathcal{L}(Q)|$.
The value $k$ reflects a tradeoff between running time and quality of the solution. As we
will see in~\secref{infocom11_simulation_2}, by picking a good value of $k$, we can obtain a lightpath
routing within a much shorter time than solving the ILP without sacrificing much in solution quality.

{\bf
{\it A Note on Complexity:}  The sets $\mathcal{C}_d$ and  $\mathcal{NC}_d$
can be constructed by enumerating all the $\choose{|E_P|}{d}$ subsets of
physical links and each of them can be classified as a cut or non-cut
in $O(|E|_L)$ time by running a
breath-first search on the logical topology.  Similarly, for each
subset $S \in \mathcal{C}_d \cup \mathcal{NC}_d$,  we can decide
whether each of its member $(i,j)$ is in $\mathcal{L}_{ij}$ and
$\mathcal{NC}_{d-1}^{st}$ by
breath-first search. Therefore, the time to compute all
$\mathcal{L}_{ij}$ is $O(\choose{|E_P|}{d}(|E_L|+d))=O(|E_P|^d|E_L|)$. Overall, the time complexity to construct the
graph $G^{'}_{P}$ is $O(|E_P|^d|E_L|)$. The 
$k$-shortest path algorithm on $G_P^{'}$ can be run 
in $O(k|V_P|(|E_P| + |V_P|\log |V_P|))$ time~\cite{yen:k_shortest}.  Therefore, the overall
time complexity of $\mathsf{REROUTE\_SP}(s,t)$ is 
$O(|E_P|^d|E_L| + k|V_P|(|E_P| + |V_P|\log |V_P|))$.
}

\subsection{Logical Topology Augmentation}
\label{sec:topology-augmentation}
The {\it Logical Topology Augmentation} problem involves finding the best way
to augment the logical topology with a single logical link, 
in order to maximize the reliability improvement.
Even though the augmentation problem has been extensively studied for single-layer networks,~\cite{AndrasAugmentation,CaiSunAugmentation,Hsu93graphaugmentation,112607,Jackson2000185}, this has not been studied before in the context of multi-layer networks. In addition to the
{\it placement} aspect of finding the end points for the new link as for the single-layer networks, there is also the {\it routing} aspect for the layered networks. This adds a new
dimension of complexity to the augmentation problem.

As it turns out, the insights from our study of the lightpath rerouting problem are largely applicable to the logical topology augmentation problem.
In the following sections,
we will first discuss the similarity between the augmentation problem and the lightpath rerouting problem,
and then develop a similar ILP formulation and approximation algorithm.

\subsubsection{Effects of a Single-Link Augmentation}
\label{sec:single-link-augment-characterize}
Similar to the rerouting problem, the new logical link chosen by the augmentation algorithm should
maximize the reduction in the number of MCLCs.
However, unlike rerouting, adding a new link
never converts a non-cut into a cross-layer cut. Therefore, in augmentation
 we only need to consider the effect of the new logical link on the existing cross-layer cuts.
 
Suppose that an initial lightpath routing is given for the 
physical topology $G_P=(V_P, E_P)$ and the logical topology $G_L=(V_L,E_L)$. 
Let $d$ be the size of the MCLC under the initial routing. 
Let $G_L^S$ be the
logical residual graph for any cross-layer cut $S$, that is, the logical subgraph in which the logical links 
 do not use any physical links in $S$. The following theorem characterizes the effect of a single-link augmentation:
The proof can be found in~\cite{kayi:thesis}.

\begin{theorem}
\label{thm:augment_cut_to_non_cut}
Let $S$ be a cross-layer cut for a lightpath routing. Augmenting the network with a new logical link
 $(s,t)$ over physical route $P$ converts a cross-layer cut $S$ into a non-cut if and only
if:

\begin{enumerate}
\item $S$ is a 2-way cross-layer cut.
\item $s$ and $t$ are disconnected in the residual graph for $S$.
\item $P$ does not use any physical links in $S$.
\end{enumerate}

\begin{proof}
The proof is the same as~\thmref{cut_to_non_cut}. The new logical link will make  the residual graph connected if and only if the above conditions are true.
\end{proof}
\end{theorem}

Note that the characterizations for augmentation (\thmref{augment_cut_to_non_cut}) and
 rerouting (\thmreftwo{cut_to_non_cut}{non_cut_to_cut}) differ only in that the conditions
 in~\thmref{non_cut_to_cut} are no longer applicable to augmentation, because augmentation
 never converts any non-cut into a cross-layer cut. Therefore, we can revise the ILP
 $\mathsf{REROUTE}$ accordingly to formulate an ILP for the augmentation problem. 

\subsubsection{ILP for Logical Topology Augmentation}
\label{sec:single-link-augment-algorithm}
We will revise the ILP $\mathsf{REROUTE}$ presented in~\secref{lightpath-rerouting} to develop the ILP for the augmentation problem. In  $\mathsf{REROUTE}$, the variables $\set{z^c=1 :c\in\mathcal{NC}_d}$ correspond to the set of non-cuts that will be converted into cuts by the rerouting, and \conreftwo{non_clc_forbidden}{non_clc_count} describe the conditions for such conversion. As previously discussed, such conversion is not applicable in augmentation and therefore these variables and constraints can be removed from the ILP. In addition, unlike rerouting where we choose from the set of existing logical links, in augmentation we can pick any two logical nodes for the new logical link. Therefore, we will replace the variable set $\set{g_{st}: (s,t)\in E_L}$ in $\mathsf{REROUTE}$ by 
$\set{g_{st}: (s,t)\in V_L\times V_L}$ and remove~\conref{one_lp}. This gives us the following ILP for augmentation:
\begin{align}
\mathsf{AUGMENT:}\quad\mbox{Maximize\ }\sum_{c\in\mathcal{C}_d}y^c, \quad\mbox{subject to:} \nonumber\\
\label{con:augment_one_lightpath}
g_{st}\leq (f_{us} + f_{tv})/2, \quad \forall (s,t)\in V_L\times V_L \\
\label{con:augment_clc_avoid}
y^c\leq \sum_{(s,t)\in  V_L\times V_L}h_{st}^c g_{st}, \quad \forall c \in\mathcal{C}_d \\
\label{con:augment_clc_connect}
y^c\leq 1-l_{ij}^cf_{ij},\quad\forall (i,j)\in E_P, \forall c \in\mathcal{C}_d \\
\label{con:augment_one_link}
\{(i,j):f_{ij}=1\}\mbox{\ forms an $(u,v)$-path in $G^{'}_P$ \ } \\
f_{ij},g_{st}\in\set{0,1}, 0\leq y^c\leq 1 \nonumber
\end{align}

Similar to the interpretation of $\mathsf{REROUTE}$, in a feasible solution to $\mathsf{AUGMENT}$, the variables $f_{ij}$ represent a path 
$u\rightarrow s\leadsto t\rightarrow v$, as described by~\conref{augment_one_link}. This corresponds to the new logical link to be added, along with its physical route.~\conref{augment_one_lightpath} ensures that $g_{st}=1$ if and only if $(s,t)$ is the new logical link selected.~\conreftwo{augment_clc_avoid}{augment_clc_connect}
describe the conditions in~\thmref{augment_cut_to_non_cut}. In particular, the variable
$y^c$ describes whether the cross-layer cut $c$ is converted into non-cut by the augmentation. Therefore, the ILP maximizes the number of such conversions, which translates to maximizing the improvement in reliability.

\subsubsection{Approximation Algorithm For Logical Topology Augmentation}
\label{sec:single-link-augment-approx}
One can also design an approximation algorithm similar to $\mathsf{REROUTE\_SP}$ introduced in~\secref{algorithm-approx} for
 the logical topology augmentation problem. We will again focus on the following question: Given a
layered network, and a new logical link $(s, t)$, find the physical route for $(s,t)$ such that the resulting number of cross-layer cuts of size $d$ is minimized. We can then apply the algorithm for this problem for every possible pair of logical nodes $s$ and $t$, to find out the new logical link that would result
in the maximum reliability improvement.

Let $d$ be the size of the MCLC of the layered network and $\mathcal{C}^{st}_d$
 be the set of 2-way cross-layer cuts of size $d$ that separate the logical 
nodes $s$ and $t$. Then by~\thmref{augment_cut_to_non_cut}, 
the set $\mathcal{L}_{ij}=\set{S\in\mathcal{C}_d^{st}:(i,j)\in S}$ represents the 
sets in $\mathcal{C}^{st}_{d}$ that will remain to be cross-layer cuts 
if the physical link $(i,j)$ is used by the $(s,t)$ path $Q$. 
We can then develop an approximation algorithm for the augmentation problem
similar to $\mathsf{REROUTE\_SP}$:

\begin{algorithm}
\caption{$\mathsf{AUGMENT\_SP}(s,t)$}
\begin{algorithmic}[1]
\STATE Construct a weighted graph on $G_P=(V_P, E_P)$,
where each edge $(i,j)$ is assigned with weight $w(i,j)=|\mathcal{L}_{ij}|$.
\STATE Run Dijkstra's algorithm to find the shortest $(s,t)$-path in the weighted graph.
\label{alg:augment_approx}
\end{algorithmic}
\end{algorithm}

Since each cross-layer cut $S$ in $\mathcal{C}^{st}_{d}$ has size $d$, there are
exactly $d$ physical links $(i,j)$ such that $S\in\mathcal{L}_{ij}$. As a 
result, $\mathsf{AUGMENT\_SP}$ is a $d$-approximation algorithm, with the
same proof as~\thmref{reroute_d_approx}.

\section{Simulation Results}
\label{sec:infocom11_simulation_2}
The single-link rerouting and augmentation methods developed in the previous section can be used as a building block
for improving the reliability of an existing layered network. For example, by iteratively rerouting the logical
links for a given lightpath routing until no further improvement is possible, we can obtain to a locally optimal solution.
In this section, we study the effectiveness of such approach via extensive simulation studies.

\subsection{Iterative Rerouting for Survivable Lightpath Routing}
\label{sec:sim_reroute}
We first apply iterative rerouting to solve the Survivable Lightpath Routing problem, whose objective is to obtain
a lightpath routing that maximizes the reliability for given physical and logical topologies. 
The Survivably Lightpath Routing has been previously studied in the literature, where the best known
 algorithmn~\cite{lee:crosslayer} is based
on an ILP fomulation that maximizes the MCLC of the network. In contrast, the objective for lightpath rerouting
algorithm is based on the lexicographical ordering of the cut vector, which captures more precisely the survivability
characteristics of the network. As we will learn from the result, using this improved objective significantly improves the
quality of the solution.

In this study, we use the NSFNET~(\figref{nsfnet}), extended with new links to raise its connectivity to 4,
 as the physical topology. For logical topologies, we generate 350 random graphs with connectivity 4,
 ranging from 6 to 12 nodes; and 13 to 38 links. For each algorithm under evaluation, we compute a
lightpath routing for each pair of physcial and logical topologies. The average reliability among the 350
lightpath routings will be presented as the performance metric.

\begin{figure}[h]
\centering
\includegraphics[width=1.6in,height=0.8in]{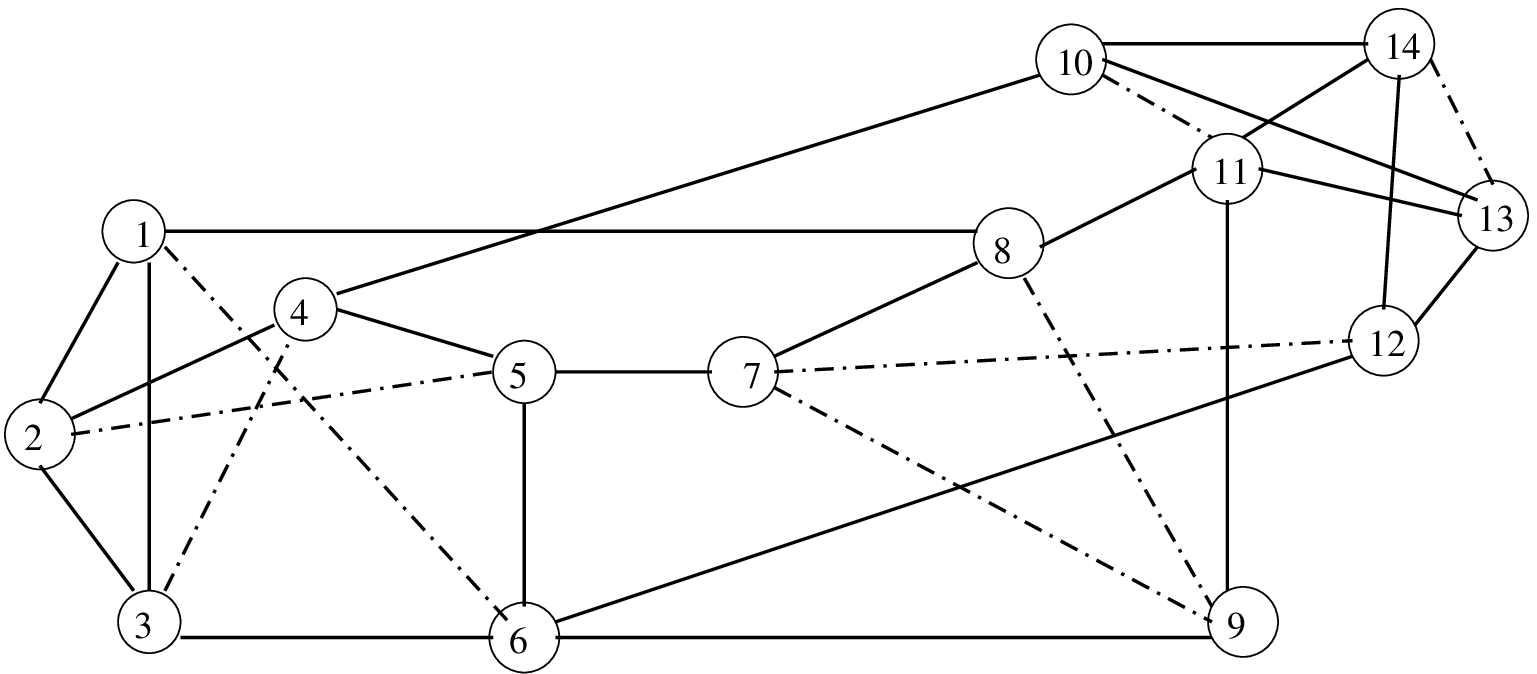}\vspace{-0.2cm}
\caption{The extended NSFNET. The dashed lines are the new links.}
\label{fig:nsfnet}
\end{figure}


We will first study the effect of the different initial lightpath routings on the reliability of the final solution. 
Next, we'll compare the performance of the rerouting algorithm variants based on ILP and the approximation algorithm. 
Throughout these studies, we also compare the solutions generated by these algorithms with the solution
generated by the best known lightpath routing algorithm in the literature, $\mathsf{MCF_{LF}}$~\cite{lee:crosslayer}
(denoted as $\mathsf{MCF}$ in this paper for simplicity), as well as the simple shortest path algorithm $\mathsf{SP}$.
 
\subsubsection{Performance of ILP-Based Rerouting}
\label{sec:reroute_vs_anneal}
{\kyb
We first evaluate the reliability performance of the
ILP-based lightpath rerouting approach
 introduced in~\secref{rerouting_ilp}, with initial lightpath routings
generated by two different algorithms:

\begin{itemize}
\item $\mathsf{RR_{SP}}$: The initial lightpath routing is generated by
the Shortest-Path algorithm $\mathsf{SP}$, which routes each lightpath with
 minimum number of physical hops.
\item $\mathsf{RR_{MCF}}$: The initial lightpath routing is generated by
the algorithm $\mathsf{MCF}$ introduced in~\cite{lee:crosslayer}.
\end{itemize}

Compared with $\mathsf{SP}$,
$\mathsf{MCF}$ provides initial lightpath routings with much
higher reliability at the expense of longer running time.
Given the initial lightpath routing, the rerouting algorithm
repeatedly solves the rerouting
ILP in~\secref{rerouting_ilp} to improve the reliability, until it reaches a local optimum.

\figref{alg_reliability} illustrates the average unreliability
 of the different algorithms. Even with initial lightpath routings
 generated by the best known lightpath routing algorithm $\mathsf{MCF}$,
the rerouting algorithm $\mathsf{RR_{MCF}}$ is able to further reduce the
unreliability of the lightpath routings. In fact, while only 50\% of the
lightpath routings generated by $\mathsf{MCF}$ has MCLC 4,
which is the connectivity of the logical topologies and is therefore
the highest MCLC value achievable, the rerouting
algorithm $\mathsf{RR_{MCF}}$ is able to archieve this maximum MCLC value
98\% of the time. This means that the lightpath rerouting approach
is able to produce lightpath routings that are much more reliable than existing
algorithms.

In addition, even though the initial lightpath routings generated by
$\mathsf{SP}$ and $\mathsf{MCF}$ differ significantly in reliability,
the iterative rerouting eliminates most of the difference. This suggests that
the rerouting approach is robust to initial routings, and
we can use a simple algorithm, such as Shortest-Path, to generate the
initial lightpath routing and rely on iterative rerouting to obtain a reliable
 solution.

\begin{figure}[h]
\centering
\includegraphics[width=3.5in,height=2in]{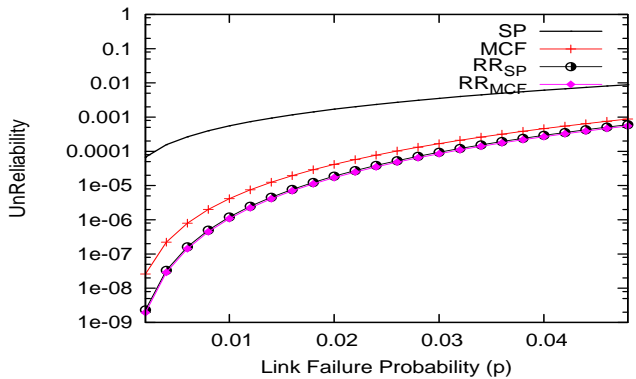}\vspace{-0.2cm}
\caption{UnReliability performance by different algorithms.}
\label{fig:alg_reliability}
\end{figure}

{\bf
~\tabref{average_path_length} shows the average physical
path length by the lightpaths generated by the different
algorithms. The higher reliability of the rerouting algorithms comes
with a cost of longer paths, as the algorithms often select the longer
physical routes in order to achieve higher reliability.
This reflects the tradeoff between the reliability and
bandwidth resource used by the lightpath routings.
}

\begin{table}[h]
\begin{center}
\begin{tabular}{|c||c|c|c|c|}  \hline
Number of& \multicolumn{4}{c||}{Average Path Length}\\  \cline{2-5}
Logical Nodes& $\mathsf{SP}$&$\mathsf{MCF}$ & $\mathsf{RR_{SP}}$ & $\mathsf{RR_{MCF}}$ \\ \hline
6 & 1.93& 2.32 & 2.60 & 2.57\\ \hline
7 & 1.90&  2.28 & 2.58 & 2.57\\ \hline
8 & 1.89& 2.30 & 2.68 & 2.62\\ \hline
9  & 1.88& 2.29 & 2.57 & 2.60\\ \hline
10 & 1.91& 2.34 & 2.68 & 2.64\\ \hline
11 & 1.90& 2.32 & 2.60 & 2.64\\ \hline
12 & 1.86& 2.22 & 2.51 & 2.51\\ \hline
\end{tabular}
\end{center}
\caption{Average path length of the Shortest-Path algorithm $\mathsf{SP}$, lightpath routing algorithm $\mathsf{MCF}$, as well as the rerouting algorithms using Shortest-Path ($\mathsf{RR_{SP}}$) and $\mathsf{MCF}$ ($\mathsf{RR_{MCF}}$) to generate the initial lightpath routings.}
\label{tab:average_path_length}\vspace{-0.3cm}
\end{table}

~\tabref{running_time_initial} shows the average running times of the rerouting
 algorithms, (not including the time to generate the initial routings), as well as the average number of rerouting iterations.
Compared with the lightpath routing algorithm $\mathsf{MCF}$, the
rerouting algorithms are able to terminate faster with a better solution.
This is because this iterative rerouting approach effectively decomposes the joint lightpath
 routing problem into simpler single-link rerouting steps, where the ILP in each step is
much smaller than the lightpath routing formulation in $\mathsf{MCF}$.

Between the two rerouting variants, $\mathsf{RR_{SP}}$ requires
more iterations than $\mathsf{RR_{MCF}}$ to reach the local optimum, because
it starts with a much less reliable initial lightpath routing.
However, the difference in total running time is less
 significant. This is because the size of the rerouting ILP formulation
is larger when the MCLC of the lightpath routing is large,
and thus takes longer to solve. In most cases, $\mathsf{RR_{SP}}$ starts with
an initial lightpath routings with a lower MCLC value. As a result,
 most of the additional
rerouting steps consist of solving the smaller ILPs to bring up the MCLC value.  Therefore, these additional steps take much shorter time.

\begin{table}[h]
\begin{center}
\begin{tabular}{|c||c|c|c||c|c|}  \hline
Number of& \multicolumn{3}{c||}{Running Time (seconds)} & \multicolumn{2}{c|}{Number of Iterations} \\  \cline{2-6}
Logical Nodes& $\mathsf{MCF}$ & $\mathsf{RR_{SP}}$ & $\mathsf{RR_{MCF}}$ & $\mathsf{RR_{SP}}$ &$\mathsf{RR_{MCF}}$  \\ \hline
6 & 1652 &265 &164 &7.0 &3.0 \\ \hline
7 &1655 &314 &257  &8.9 & 4.2 \\ \hline
8 &1732 &500 &365  &10.3 & 5.0 \\ \hline
9  & 1838 &745 &525  &11.6 & 6.2 \\ \hline
10 & 2032 &1238 &824  &14.1 & 7.3 \\ \hline
11 & 2219 &1389 &1280   &14.0 & 8.0 \\ \hline
12 & 2716 &1268 &1104   &14.1 & 8.2 \\ \hline
\end{tabular}
\end{center}
\caption{Running times of the lightpath routing algorithm $\mathsf{MCF}$, as well as the rerouting algorithms using Shortest-Path ($\mathsf{RR_{SP}}$) and $\mathsf{MCF}$ ($\mathsf{RR_{MCF}}$) to generate the initial lightpath routings; and the number of iterations of the rerouting algorithms.}
\label{tab:running_time_initial}\vspace{-0.3cm}
\end{table}

\subsubsection{Performance of Approximation Algorithm}
\label{sec:reroute_approx}
Next, we compare the performance of the
approximation algorithm introduced
in~\secref{algorithm-approx} with the ILP counterpart. As discussed, the
approximation algorithm is based on the $k$-shortest-path algorithm,
where the parameter $k$ reflects a tradeoff
between running time and reliability performance. We evaluate this
algorithm, $\mathsf{Shortest_k}$, with $k=$1, 10 and 100.

We use the lightpath routings generated by the Shortest Path algorithm as the initial routings.
\figref{rerouting_reliability_approx} shows the average unreliability of the lightpath routings produced by the
algorithms. While $\mathsf{Shortest_1}$ brings in the majority of the improvement, increasing the
value of $k$ is able to further reduce the unreliability. In particular,
when $k=100$, the approximation algorithm performs almost as well as solving
the rerouting ILP.

\begin{figure}[h]
\centering
\includegraphics[width=3.5in,height=2in]{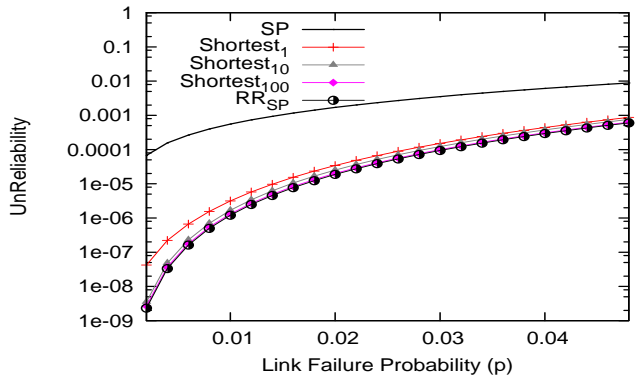}\vspace{-0.2cm}
\caption{Lightpath rerouting: performance of approximation algorithm.}
\label{fig:rerouting_reliability_approx}\vspace{-0.2cm}
\end{figure}

{\kyb
~\tabref{running_time_approx} compares the running time of the algorithms.
As shown in the table, the approximation algorithms are significantly
faster than the ILP-based algorithm.
This suggests that the
approximation algorithm is promising rerouting alternative to the ILP for
improving the reliability of large networks, without the 
need to solve complex mathematical programs.
}

\begin{table}[h]
\centering
\begin{tabular}{|c||c|c|c|c|}  \hline
Number of&  \multicolumn{4}{c|}{Running Time (seconds)} \\  \cline{2-5}
Logical Nodes& $\mathsf{RR_{SP}}$ &$\mathsf{Shortest_1}$ & $\mathsf{Shortest_{10}}$ & $\mathsf{Shortest_{100}}$ \\ \hline
6 &265 &12 &14 &24  \\ \hline
7 &314 &20 &26 &43  \\ \hline
8 &500 &32 &43 &79  \\ \hline
9 &745 &45 &55 &123  \\ \hline
10 &1238 &68 &91 &199 \\ \hline
11 &1389 &83 &104 &254  \\ \hline
12 &1268 &113 &135 &344  \\ \hline
\end{tabular}
\caption{Running times of the rerouting algorithms based on ILP ($\mathsf{RR_{SP}}$) and $k$-shortest paths.}
\label{tab:running_time_approx}\vspace{-0.5cm}
\end{table}

\subsection{Effects of Logical Topology Augmentation}
\label{sec:augment_sim}
Next, we study the effect of augmenting the logical topology on the network reliability. 
We study a 10-node and a 14-node logical ring on the augmented NFSNET, 
as shown in~\figref{augment_ring_nsfnet}, and incrementally augment the rings to
 study the reliability improvement from the addition of new logical links.
\begin{figure}[h]
\centering
\subfigure[10 Node Logical Ring]{
\includegraphics[height=0.15\textwidth]{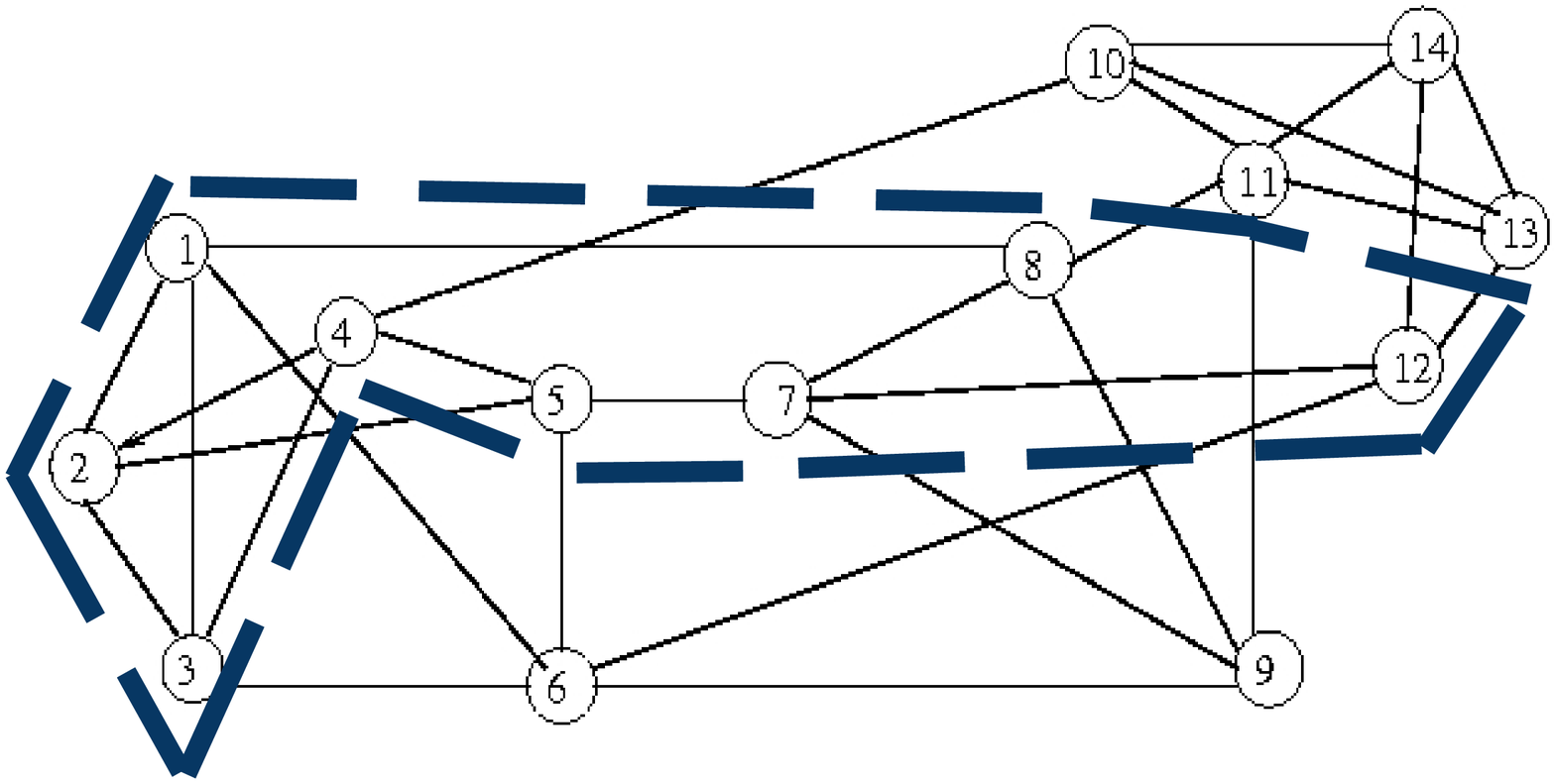}}
\subfigure[14 Node Logical Ring]{
\includegraphics[height=0.15\textwidth]{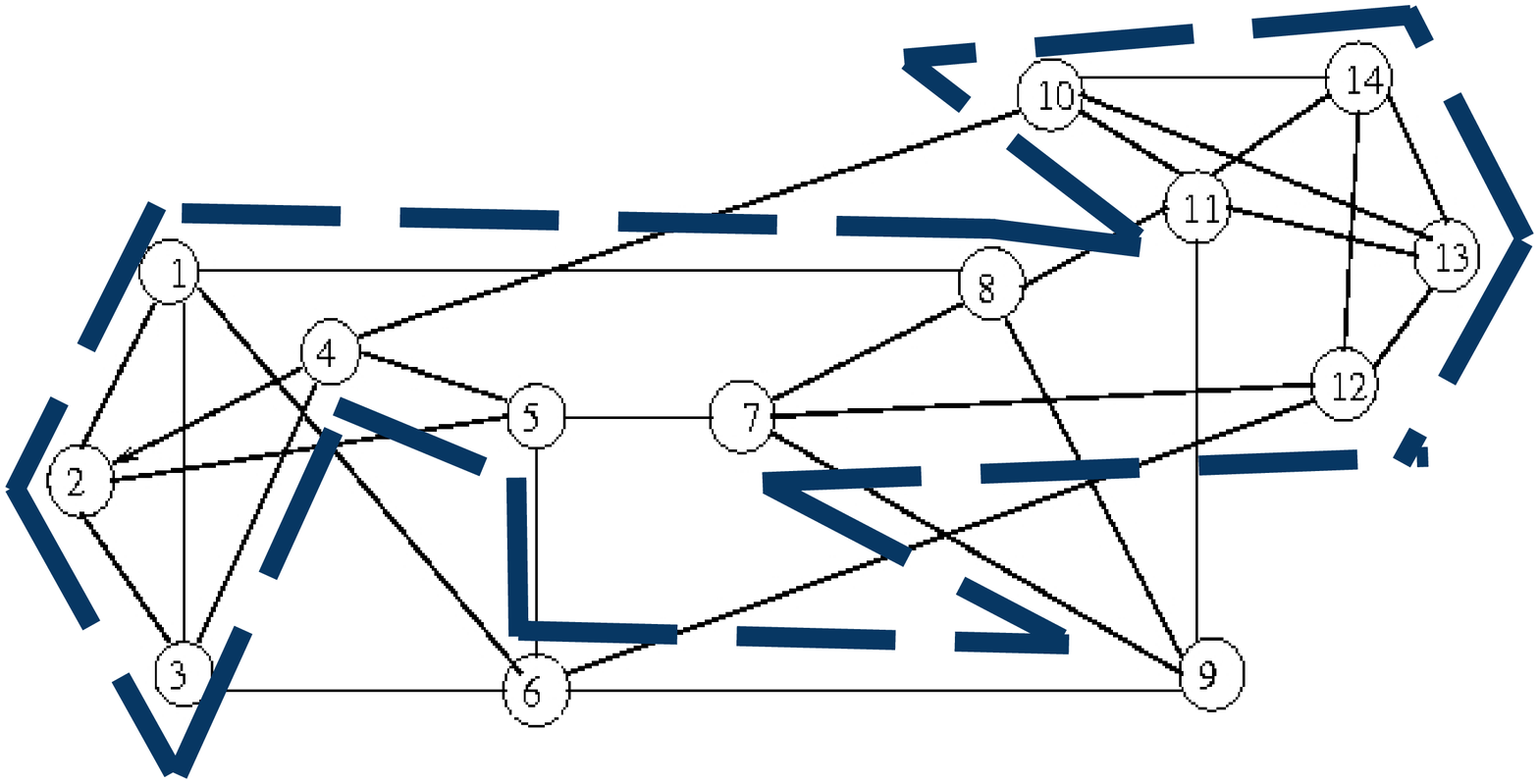}}
\caption{Logical rings on extended NSFNET.}
\label{fig:augment_ring_nsfnet}
\end{figure}

The cross-layer reliability of the networks
after each augmentation step is shown in~\figref{augment_ring}.
With link failure probability $p=0.01$, the unreliability 
 declines as we add more logical links to the rings. The key observation from these figures
is that the improvement in reliability is most prominent when the augmentation increases the MCLC
of the network. This suggests that networks with a small number of MCLCs have a greater
potential to significantly improve the reliability by augmentation, as it is more likely
to improve their MCLC values by a small number of new logical links.

In the case where the additional link does not cause an MCLC increase, 
the marginal reliability improvement decreases with the current MCLC value. 
This means that augmentation is more effective when the MCLC value is lower.

\begin{figure}[h]
\centering
\subfigure[10 Node Logical Ring]{
\includegraphics[width=0.4\textwidth]{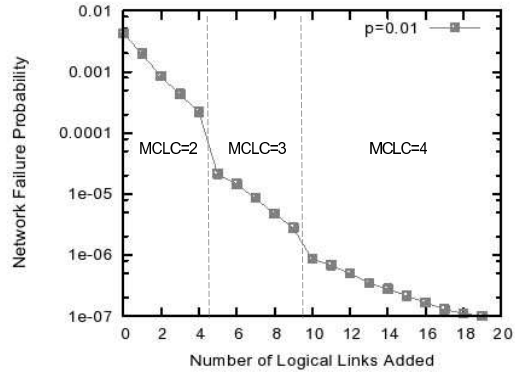}}
\subfigure[14 Node Logical Ring]{
\includegraphics[width=0.4\textwidth]{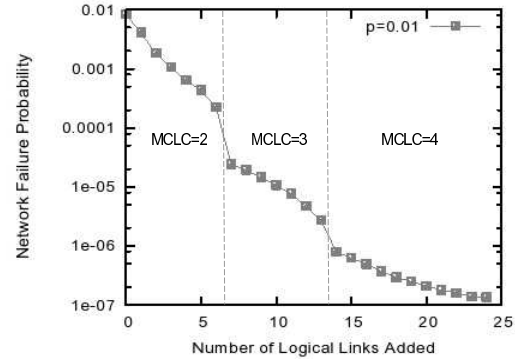}}\vspace{-0.2cm}
\caption{Impact on reliability by augmenting logical rings.}
\label{fig:augment_ring}\vspace{-0.3cm}
\end{figure}

\subsection{Case Study: A Real-World IP-over-WDM Topology}
\label{sec:real-world}
Finally, we evaluate the performance
of the rerouting and augmentation algorithms on a large layered network based on a real-world IP-over-WDM
network. The physical and logical topologies, shown in~\figref{qwest_ip_over_wdm},
are constructed based on the network maps available from Qwest 
Communications~\cite{qwest}. Both the physical and logical topologies are extended with new links so that 
the graphs have connectivity 4. 
The physical topology has 39 nodes and 72 links, and the logical topology has 
20 nodes and 101 links.

\begin{figure}
\centering

\subfigure[WDM (physical) network.]{
\includegraphics[angle=0,width=0.3\textwidth]{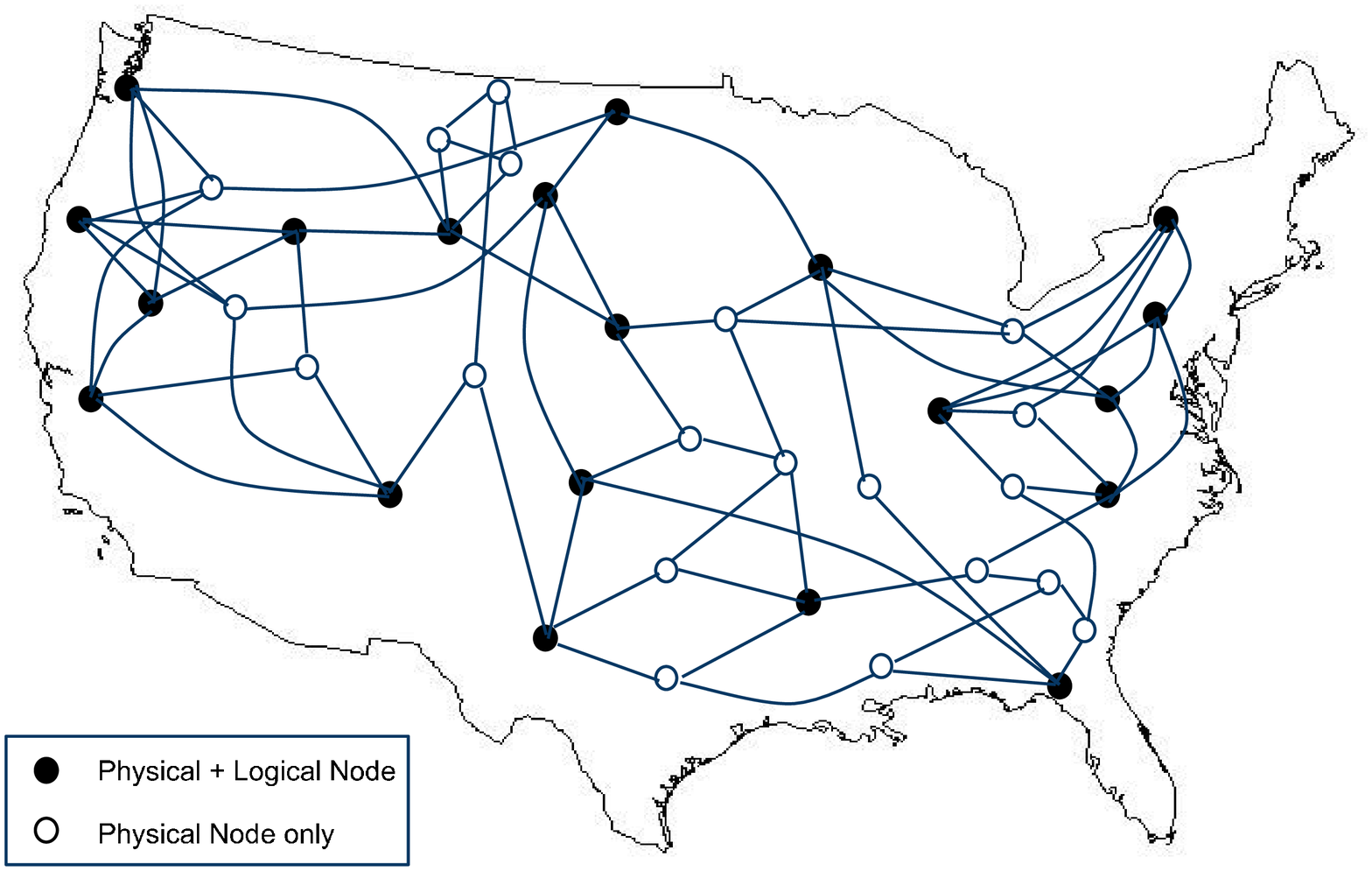}}
\subfigure[IP/MPLS (logical) network. The numbers indicate the number of 
parallel logical links between the logical nodes.]{
\includegraphics[angle=0,width=0.3\textwidth]{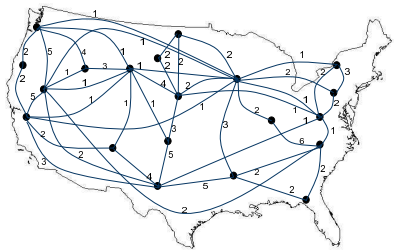}}\vspace{-0.2cm}
\caption{Physical and logical topologies.}
\label{fig:qwest_ip_over_wdm}\vspace{-0.3cm}
\end{figure}

The study on larger networks allows us to reevaluate the performance of
the lightpath algorithms, both in terms of scalability and
solution quality. In this study, we run the following lightpath
routing algorithms and compare their solutions:

\begin{enumerate}
\item $\mathsf{MCF}$: The multi-commodity flow algorithm introduced in~\cite{lee:crosslayer}. As in~\secref{sim_reroute},
the algorithm is evaluated as the performance baseline.
\item $\mathsf{REROUTE}$: The iterative lightpath rerouting algorithm, based on
the $k$-shortest path algorithm presented in~\secref{algorithm-approx}, where $k$ is
set to 5000 in our experiment.
\item $\mathsf{AUGMENT_n}$: The logical topology augmentation algorithm, based on
the $k$-shortest path algorithm presented in~\secref{single-link-augment-approx}, where
$k$ is set to 5000 in our experiment. The augmentation algorithm is run successively to add $n$ new
edges on the lightpath routing given by $\mathsf{REROUTE}$, where $n=1,\ldots,9$.
\end{enumerate}

The MCLC values and the number of MCLCs of the lightpath routings generated by 
each algorithm are shown in~\tabref{qwest_mclc}. These numbers are compared against
the lower bound, which is computed by counting the number of minimum sized 
physical fiber sets whose removal will {\it physically} disconnect
some logical nodes. These sets of physical links are cross-layer cuts regardless of
the lightpath routing, and therefore will provide a lower bound on the number of
MCLCs.

It was observed in~\cite{lee:crosslayer} that the survivability performance of the 
multi-commodity flow formulation $\mathsf{MCF}$ declines as the network 
size increases. In this case, the solution produced by the algorithm only has MCLC value 2.
On the other hand, the
rerouting algorithm $\mathsf{REROUTE}$ continues to produce a lightpath routing with the maximum possible MCLC
value 4.
Augmenting the logical topology can further improve the 
reliability of the layered network by reducing the number of MCLCs, 
though the incremental effect declines as more
logical links are added to the network. The number of MCLCs hits the
lower bound when the logical topology is augmented with 9 additional logical 
links.

~\figref{qwest_rel} compares the algorithms in terms of the cross-layer reliability
in the low failure probability regime. 
As suggested by~\tabref{qwest_mclc}, the iterative algorithms  
achieve significantly higher reliability than the existing algorithm $\mathsf{MCF}$ (by about 3 orders of magnitude).
In particular, the majority of the improvement is achieved by the lightpath rerouting algorithm
$\mathsf{REROUTE}$. This is because the lightpath rerouting method alone is able to achieve the maximum MCLC value.
As we observed in~\secref{augment_sim},
adding logical links is more effective only if the new links can raise the MCLC of the network.
In other words,
even without adding new logical links, we can obtain a near optimal solution by
improving the existing lightpath routing via the iterative rerouting method.

\begin{table}[h]
\begin{center}
\begin{tabular}{|c||c|c|}  \hline
Algorithm& MCLC& Number of MCLCs \\ \hline
$\mathsf{MCF}$&2&5 \\ \hline
$\mathsf{REROUTE}$&4&216 \\ \hline
$\mathsf{AUGMENT_1}$&4&84 \\ \hline
$\mathsf{AUGMENT_3}$&4&34 \\ \hline
$\mathsf{AUGMENT_5}$&4&25 \\ \hline
$\mathsf{AUGMENT_9}$&4&20 \\ \hline
\hline
{\bf Lower Bound}&{\bf 4}&{\bf 20} \\ \hline
\end{tabular}
\end{center}
\caption{MCLC values and MCLC counts of different lightpath routings.
The lightpath routing on a logical topology augmented with $n$ new logical 
links is denoted by $\mathsf{AUGMENT_n}$.}
\label{tab:qwest_mclc}\vspace{-0.5cm}
\end{table}

\begin{figure}[h]
\centering
\includegraphics[width=3.5in,height=2in]{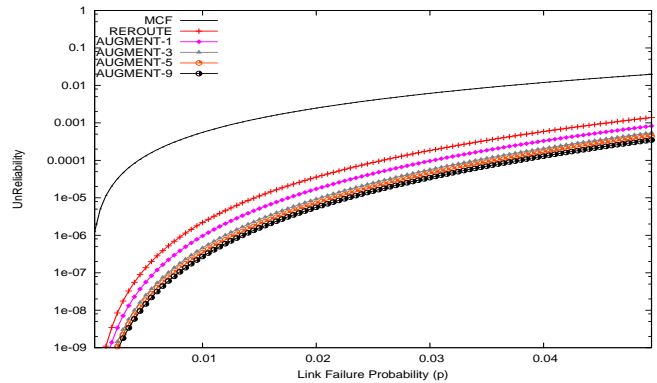}
\caption{Unreliability of different lightpath routings.}
\label{fig:qwest_rel}\vspace{-0.5cm}
\end{figure}

\section{High Failure Probability Regime}\label{sec:high-regime}
As discussed in Section \ref{sec:intro}, natural disasters or physical attacks can lead to widespread network link failures.  While such events may be extremely rare, certain networks that are of critical importance to national security and our day to day lives may need to be designed so that they can withstand such rare events.  Moreover, certain ``specialized'' networks, such as those onboard an aircraft or a ship may need to be designed to withstand very high link failure probabilities that  result from a catastrophic failure event (e.g., well over 50\% link failures) \cite{We2004}.  In this section, we briefly discuss network design in this high failure probability scenario.  

In Section \ref{sub:low-regime}, we showed that when $p$ is small, it is important to minimize the number of small cuts. Analogously, for large $p$, large cuts are dominant, and hence, minimizing the number of large cuts would result in maximum reliability. In other words, the cut vector should be minimized for large cuts for better reliability in the high failure probability regime. Similar to the case of low probability regime, we define the following vector of partial sums:
\[
\overleftarrow{{\bf N}}=\left[\sum\limits_{i=m-k}^mN_i|k=0,...,m\right].
\]
The vector $\overleftarrow{{\bf M}}$ is defined similarly. Note that the $i$-th element $\overleftarrow{N}_i$ is the total number of cross-layer cuts of size at least $m-i$. We will use these vectors to develop the conditions that incrementally include larger cuts and characterize the probability regime where one lightpath routing is more reliable than any other for large $p$.

First, the $k$-colexicographical ordering (an analogy to $k$-lexicographical ordering in Section \ref{sec:infocom11_regime_extension}) is defined as follows:

\begin{definition}
Lightpath routing $1$ is said to be  $k$-colexicographically smaller than lightpath routing $2$ if
\begin{align*}
 k=\max\set{j:\overleftarrow{N}_{i}\leq\overleftarrow{M}_{i},\quad\forall i> c-j} \mbox{\ and $k\geq1$},
\end{align*}
where $c$ is the position of last element where the two cut vectors differ.
\end{definition}

In contrast to the $k$-lexicographical ordering, this colexicographical ordering starts from the largest cuts, and incrementally includes the smaller cuts. The following result is similar to Theorem \ref{thm:partial_sum_dominance}:
\begin{theorem}\label{thm:partial_sum_dominance-h}
Given two vectors {\bf N}=$[N_i|i=0,\ldots,m]$ and
{\bf M}=$[M_i|i=0,\ldots,m]$. For any $j$,
let $\overleftarrow{\Delta}_j=\sum\limits_{i=m-j}^m(M_i-N_i)$ and $\overleftarrow{\delta}_j = \max\limits_{0\leq i\leq m-j-1}\set{\frac{N_i-M_i}{\choose{m}{i}}}$.
If {\bf N} is $k$-colexicographically smaller than {\bf M}, then
\[
F_1(p)\leq F_2(p) \textrm{ for } p\geq p_0^h=1-\max\set{0.5, \min\limits_{c\leq j\leq c+k-1}C_j},
\] 
where $c=\min\set{i:N_{m-i}<M_{m-i}}$ and
\begin{align*}
C_j=\left\{\begin{array}{l l}
0.5,&\mbox{if $j=m$} \nonumber \\
1-\frac{1}{\frac{m}{j+1} + \overleftarrow{\delta}_j\choose{m}{j+1}/\overleftarrow{\Delta}_j},& \mbox{otherwise.} \nonumber \\
\end{array}
\right.
\end{align*}
\end{theorem}
\begin{proof}
See Appendix \ref{app:high_regime_partial}.
\end{proof}

The following corollary is analogous to~\corref{p-leq-.5} for the high failure regime:
\begin{corollary}\label{cor:p-geq-.5}
If $\overleftarrow{N}_j\leq\overleftarrow{M}_j$ for all $j=0,\ldots,m$,
then routing $1$ is at least as reliable as routing $2$ for $p\geq0.5$, i.e., $F_1(p)\leq F_2(p)$ for $p\geq0.5$.
\end{corollary}
Combining~\correftwo{p-leq-.5}{p-geq-.5} gives a condition for uniformly optimal lightpath routing:
\begin{corollary}\label{cor:local->uniform}
If $\overrightarrow{N}_j\leq\overrightarrow{M}_j$ and $\overleftarrow{N}_j\leq\overleftarrow{M}_j$ for all $j=0,\ldots,m$,
then lightpath routing 1 is uniformly optimal.
\end{corollary}

\thmreftwo{partial_sum_dominance}{partial_sum_dominance-h} provide a single optimality regime expression for lightpath routings that exhibit different degrees of dominance. Note that the conditions of (co)lexicographical ordering in Corollary \ref{cor:local->uniform} are satisfied by the uniform optimality condition $N_i\leq M_i,\forall i$ given in~\thmref{uniform_optimal_dominance}. Therefore, this unified theorem allows for a broader class of uniformly optimal lightpath routings.

More importantly, Theorem \ref{thm:partial_sum_dominance-h} can be used to derive practical conditions for optimal lightpath routings in the high failure probability regime. We begin with the following definitions:
\begin{definition}
Consider two lightpath routings $1$ and $2$. Routing $1$ is said to be \emph{more reliable} than routing $2$ \emph{in the high failure probability regime} if there exists a number $p_0<1$ such that the reliability of routing $1$ is higher than that of routing $2$ for $p_0<p<1$.
\end{definition}
\begin{definition}
A \emph{cross-layer spanning tree} is a \emph{minimal} set of fibers whose survival keeps the logical network connected. Hence, if $T$ is a cross-layer spanning tree, then the survival of just $T\setminus\{(i,j)\}$ renders the logical network disconnected for any fiber $(i,j)\in T$.
\end{definition}

Note that the cross-layer spanning tree is a generalization of the single-layer spanning tree. However, unlike a single-layer graph where all spanning trees have the same size, in a layered graph, spanning trees can have different sizes. Thus, we define a \emph{Min Cross Layer Spanning Tree (MCLST)} as a spanning tree with minimum number of physical links. 

In the high failure probability regime, it is likely that there are a large number of failures. Hence, the MCLST is an important parameter in the high failure probability regime because logical networks with small MCLST may remain connected even if only a small number of physical links survive. This intuition together with Theorem \ref{thm:partial_sum_dominance-h} leads to practical conditions for optimal routings in the high failure probability regime. 

Note that in the failure polynomial, $N_i\leq\binom{m}{i}$. Let $m-c$ be the size of MCLST. Then, $c$ is the largest $i$ such that $N_i<\binom{m}{i}$, and we have $N_i=\binom{m}{i},\forall i>c$, meaning that more than $c$ failures would always disconnect the logical network. Let $m-c_j$ be the size of MCLST for routing $j$. It is obvious that if $c_1>c_2$ or $c_1=c_2$ \& $N_{c_1}<M_{c_2}$, then routing $1$ is $k$-colexicographically smaller than routing $2$. This observation leads to the corollaries similar to the low regime case:
\begin{corollary}\label{thm:high-regime1}
If $c_1>c_2$ (i.e., if routing $1$ has smaller MCLST than routing $2$), then routing $1$ is more reliable than routing $2$ in the high failure probability regime.
\end{corollary}
\begin{corollary}\label{thm:high-regime2}
If $c_1=c_2$ and $N_{c_1}<M_{c_2}$ (i.e., routings $1$ and $2$ have the same size of MCLST, but routing $1$ has more MCLSTs), then routing $1$ is more reliable than routing $2$ in the high failure probability regime.
\end{corollary}

%

Therefore, for reliability maximization in the high failure probability regime, it is desirable to find a lightpath routing that minimizes the size of MCLST and maximizes the number of MCLSTs. This observation is similar to the
 single-layer setting where maximizing the number of spanning trees maximizes the reliability for large $p$ \cite{bauer:validity}. The major difference in the multi-layer case is that, since spanning trees may have
different sizes, minimizing the {\it size} of the Min Cross-Layer Spanning Tree becomes the primary objective. Moreover, computing the size of the MCLST is NP-hard \cite{lee:reliability}, and therefore, designing a lightpath routing that minimizes the MCLST is likely to be a difficult problem. We developed an ILP-based algorithm that finds a lightpath routing with minimum-size MCLST, and its details can be found in Appendix \ref{app:algorithm_high_regime}.

\section{Conclusion}
We studied the reliability maximization problem in layered networks with random link failures. We introduced the notion of lexicographical ordering for lightpath routings, and fully identified optimization criteria for maximum reliability in the low failure probability regime. In particular, we showed that a lightpath routing with the maximum size of Min Cross Layer Cut (MCLC) and the minimum number of MCLCs is most reliable in the low failure probability regime. Based on this insight, we developed a novel lightpath rerouting approach to design reliable layered networks for the low failure probability regime. By incrementally improving the lightpath routing, this rerouting approach is able to achieve a locally optimal solution. Our simulation results show that the rerouting algorithms developed in this paper are able to produce much more reliable lightpath routings than existing algorithms (by about 3 orders of magnitude in real IP-over-WDM Network), and are more scalable to large networks. Using the optimization criteria, we also developed logical topology augmentation algorithms that can further improve the reliability of a given layered network.

We also showed that the high failure probability regime requires different optimization criteria that a routing with the minimum size of Min Cross Layer Spanning Tree (MCLST) and the maximum number of MCLSTs maximizes reliability. Our results in the high failure probability regime lay the foundation for the design of networks facing increased concern about large scale failures due to natural disasters or attacks.




\appendices
\section{Proof of Theorem \ref{thm:partial_sum_dominance}}\label{app:unifying-thm}

We first prove the following lemma.
\begin{lemma}\label{lem:induction-partial-sum}
If vector {\bf N} is $k$-lexicographically smaller than vector {\bf M}, then for all $j \leq d+k-1$, where $d=\min\set{d:N_d< M_d}$, and for $0\leq p\leq 0.5$,
\begin{align}
\sum_{i=0}^j(M_i-N_i)p^i(1-p)^{m-i}\geq\overrightarrow{\Delta}_jp^j(1-p)^{m-j}.
\label{eqn:delta_bound}
\end{align}

\begin{proof}
We prove, by induction on $j$, that (\ref{eqn:delta_bound}) holds for all $j\leq d+k-1$. Since $\overrightarrow{\Delta}_0=M_0-N_0$, we have for $j=0$,
\begin{align*}
\sum_{i=0}^j(M_i-N_i)p^i(1-p)^{m-i}&=\overrightarrow{\Delta}_0(1-p)^m
\end{align*}
Therefore, (\ref{eqn:delta_bound}) holds for $j=0$. Now suppose (\ref{eqn:delta_bound}) holds for all $i\leq j$ for some $j< d+k-1$. Then, we have:

\begin{align*}
&\sum_{i=0}^{j+1}(M_i-N_i)p^i(1-p)^{m-i} \\
=&\sum_{i=0}^{j}(M_i-N_i)p^i(1-p)^{m-i} + (M_{j+1}-N_{j+1})p^{j+1}(1-p)^{m-(j+1)} \\
\geq&\overrightarrow{\Delta}_jp^j(1-p)^{m-j}+(M_{j+1}-N_{j+1})p^{j+1}(1-p)^{m-(j+1)}\\
\geq&\overrightarrow{\Delta}_jp^{j+1}(1-p)^{m-(j+1)}+(M_{j+1}-N_{j+1})p^{j+1}(1-p)^{m-(j+1)}\\
=&\overrightarrow{\Delta}_{j+1}p^{j+1}(1-p)^{m-(j+1)},
\end{align*}
where the first inequality is due to the induction hypothesis, and the second inequality is because $p/(1-p)\leq1$. Therefore, by induction,  (\ref{eqn:delta_bound}) is true for all $j\leq k$.
\end{proof}
\end{lemma}

\begin{lemma}\label{lem:one_partial_sum}
Given a fixed $k$, if $\overrightarrow{\Delta}_i\geq 0$ for all
$i\leq d+k-1$, then for any $d\leq j\leq d+k-1$:
\begin{align*}
F_1(p)\leq F_2(p),
\end{align*}
for $0\leq p\leq \min\set{0.5, B_j}$, where:
\begin{align*}
B_j=\left\{\begin{array}{l l}
0.5,&\mbox{if $j=m$} \nonumber \\
\frac{1}{\frac{m}{j+1} + \overrightarrow{\delta}_j\choose{m}{j+1}/\overrightarrow{\Delta}_j},& \mbox{otherwise.} \nonumber \\
\end{array}
\right.
\end{align*}
\end{lemma}

\begin{proof}
First, note that by definition of $\overrightarrow{\delta}_j$, for any $i\geq j$:
\begin{equation}
\overrightarrow{\delta}_j\choose{m}{i}\geq N_i-M_i.
\label{eqn:delta_j}
\end{equation}
If $k=m-d+1$, then~\lemref{induction-partial-sum} implies that, for $p\leq 0.5$:
\begin{align*}
\sum_{i=0}^m(M_i-N_i)p^i(1-p)^{m-i}
&\geq\overrightarrow{\Delta}_mp^m\geq 0.
\end{align*}
Therefore, the lemma is true for $k=m-d+1$.
Now suppose $k<m-d+1$. If $\overrightarrow{\delta}_j\leq 0$ for some $j\leq k$,
this implies for any $d+k\leq l\leq m$:
\begin{align*}
\overrightarrow{\Delta}_l &= \overrightarrow{\Delta}_{d+k-1} + \sum_{i=d+k}^l(M_i-N_i) \\
&\geq \overrightarrow{\Delta}_{d+k-1} - \sum_{i=d+k}^l\overrightarrow{\delta}_j\choose{m}{i}\geq0,
\end{align*}
where the first inequality is due to (\ref{eqn:delta_j}). The second inequality is due to the fact that $\overrightarrow{\delta}_j\leq 0$, and that $\overrightarrow{\Delta}_{d+k-1}\geq 0$, since {\bf N} is $k$-lexicographically smaller than {\bf M}.
Therefore, in this case, the vector {\bf N} is also $(m-d+1)$-lexicographically smaller than {\bf M}, and the lemma is true as proved above. Therefore, in the rest of the
proof, we assume that $\overrightarrow{\delta}_j>0$.

Since $p\leq 0.5$ and $\overrightarrow{\Delta}_i\geq 0$ for all
$i\leq d+k-1$, by~\lemref{induction-partial-sum} we have, for all $j\leq d+k-1$:
\begin{align}
\sum\limits_{i=0}^{j}(M_{i}-N_{i}){p}^{i}(1-p)^{m-i}
\geq\overrightarrow{\Delta}_jp^{j}(1-p)^{m-j}.
\label{eqn:induction-j}
\end{align}

Next, we will use the following result to bound the tail probability of the Binomial distribution:
\begin{lemma}[\cite{feller:probability}]
\label{lem:tail-prob}
For $r>mp$,
\begin{align*}
\sum_{i=r}^m\binom{m}{i}p^i(1-p)^{m-i}\leq\binom{m}{r}p^r(1-p)^{m-r}\cdot\frac{r(1-p)}{r-mp}.
\end{align*}
\end{lemma}
Therefore, since $p\leq\frac{1}{\frac{m}{j+1} + \overrightarrow{\delta}_j\choose{m}{j+1}/\overrightarrow{\Delta}_j}<\frac{j+1}{m}$, by Lemma \ref{lem:tail-prob}, we have:

\begin{align}
&\sum_{i=j+1}^m\binom{m}{i}p^i(1-p)^{m-i}\\
&\leq\binom{m}{j+1}p^{j+1}(1-p)^{m-(j+1)}\cdot\frac{(j+1)(1-p)}{j+1-mp} \nonumber\\
&=\binom{m}{j+1}p^{j}(1-p)^{m-j}\cdot\frac{(j+1)p}{j+1-mp}.
\label{eqn:tail-bound}
\end{align}
In addition, since
$p\leq\frac{1}{\frac{m}{j+1} + \overrightarrow{\delta}_j\choose{m}{j+1}/\overrightarrow{\Delta}_j}$, we have:
\begin{align}
\frac{(j+1)p}{j+1-mp}
&=\frac{1}{\frac{1}{p}-\frac{m}{j+1}} \nonumber\\
&\leq \frac{1}{\frac{\overrightarrow{\delta}_j\choose{m}{j+1}}{\overrightarrow{\Delta}_j} + \frac{m}{j+1} - \frac{m}{j+1}}=\frac{\overrightarrow{\Delta}_j}{\overrightarrow{\delta}_j\choose{m}{j+1}}.
\label{eqn:p_expression}
\end{align}
{\allowdisplaybreaks
It follows that:
\begin{align*}
&\sum\limits_{i=0}^{m}(M_{i}-N_{i}){p}^{i}(1-p)^{m-i} \\
=&\sum\limits_{i=0}^{j}(M_{i}-N_{i}){p}^{i}(1-p)^{m-i} + \sum\limits_{i=j+1}^{m}(M_{i}-N_{i}){p}^{i}(1-p)^{m-i} \\
\geq&\sum\limits_{i=0}^{j}(M_{i}-N_{i}){p}^{i}(1-p)^{m-i} - \sum\limits_{i=j+1}^{m}\overrightarrow{\delta}_j\choose{m}{i}{p}^{i}(1-p)^{m-i}\\
\geq& \overrightarrow{\Delta}_jp^{j}(1-p)^{m-j}-\overrightarrow{\delta}_j\binom{m}{j+1}p^{j}(1-p)^{m-j}\cdot\frac{(j+1)p}{j+1-mp}\\
=& p^{j}(1-p)^{m-j}\overrightarrow{\delta}_j\left(\frac{\overrightarrow{\Delta}_j}{\overrightarrow{\delta}_j} - \binom{m}{j+1}\cdot\frac{(j+1)p}{j+1-mp}\right) \\
\geq& p^{j}(1-p)^{m-j}\overrightarrow{\delta}_j\left(\frac{\overrightarrow{\Delta}_j}{\overrightarrow{\delta}_j} - \binom{m}{j+1}\cdot\frac{\overrightarrow{\Delta}_j}{\overrightarrow{\delta}_j\choose{m}{j+1}}\right)=0.
\end{align*}
}
The first inequality is due to (\ref{eqn:delta_j}), the second inequality is due to (\ref{eqn:induction-j}) and (\ref{eqn:tail-bound}), and the last inequality is due to (\ref{eqn:p_expression}).
\end{proof}

As a result of~\lemref{one_partial_sum}, we can pick the $d\leq j\leq d+k-1$ such that $B_j$
is maximized to obtain the largest upper bound for $p$, and~\thmref{partial_sum_dominance} follows.

\section{Proof of~\thmref{reroute_d_approx}}
\label{app:rerouting-approx}

Given any $(s,t)$ path $Q$, define
$\mathcal{L}(Q)=\cup_{(i,j)\in Q}L_{ij}$, it follows that
$N_d(Q)=|\mathcal{L}(Q)| + |\mathcal{C}_d^{st}|= |\mathcal{L}(Q)| + K$, where $K=|\mathcal{C}_d^{st}|$
is a constant. In addition, let $w(Q)$ be the total weight sum of the path $Q$ in the weighted graph
constructed by $\mathsf{REROUTE\_SP}(s,t)$.

Since each set of physical links $S\in \mathcal{L}(Q)$ has size $d$, we have
$|\set{(i,j):S\in \mathcal{L}_{ij}}|\leq d$, which implies:
\begin{align}
w(Q)&=\sum\limits_{(i,j)\in Q}|\mathcal{L}_{ij}| \nonumber \\
&\leq\sum\limits_{S\in\mathcal{L}(Q)}{|\set{(i,j):S\in \mathcal{L}_{ij}}|} \nonumber \\
&\leq d\cdot |\mathcal{L}(Q)|=d\cdot (N_d(Q) - K)
\label{eqn:weight_cardinality2}
\end{align}

Now, since $Q^{\mathsf{SP}}$ is the minimum weight $(s,t)$ path in the graph, it follows that:
\begin{align*}
N_d(Q^{\mathsf{SP}})
&= |\mathcal{L}(Q^{\mathsf{SP}})| + K\\
&\leq w(Q^{\mathsf{SP}}) + K\\
&\leq w(Q^{*}) + K\\
&\leq d\cdot(N_d(Q^{*}) - K) +K,\quad\mbox{by~\eqnref{weight_cardinality2}} \\
&\leq d\cdot N_d(Q^{*}).
\end{align*}

\section{Proof of~\thmref{partial_sum_dominance-h}}
\label{app:high_regime_partial}

Let $N^{'}_{i}=N_{m-i}$ and $M^{'}_{i}=M_{m-i}$, for $i=0,\ldots,m$;
and let $\overrightarrow{N_k^{'}} = \sum_{i=0}^k{N_i^{'}}$ and $\overrightarrow{M_k^{'}}=\sum_{i=0}^k{M_i^{'}}$. It follows that the vector
$\overrightarrow{N^{'}}:=\left[\overrightarrow{N^{'}_i}|i=0,\ldots,m\right]$
is $k$-lexicographically smaller than the vector $\overrightarrow{M^{'}}:=\left[\overrightarrow{M^{'}_i}|i=0,\ldots,m\right]$.
Let $q=1-p$. Then, by ~\thmref{partial_sum_dominance},
\begin{align*}
\sum_{i=0}^{m}(M_i-N_i)p^i(1-p)^{m-i}
&=\sum_{i=0}^{m}(M^{'}_i-N^{'}_i)q^i(1-q)^{m-i}\geq 0,
\end{align*}
for $q\leq\min\set{0.5,\max\limits_{d\leq j\leq d+k-1}B_j}$, where:
\begin{align*}
B_j&=\left\{\begin{array}{l l}
0.5,&\mbox{if $j=m$} \nonumber \\
\frac{1}{\frac{m}{j+1} + \overrightarrow{\delta}^{'}_j\choose{m}{j+1}/\overrightarrow{\Delta}^{'}_j},& \mbox{otherwise} \nonumber \\
\end{array}
\right.
\end{align*}
In the above expression, we have:
\begin{align*}
\overrightarrow{\Delta}^{'}_j=\sum\limits_{i=0}^j{M_i^{'}-N_i^{'}}=\overleftarrow{\Delta}_j, \quad\mbox{and} \\
\overrightarrow{\delta}^{'}_j=\max\limits_{j+1\leq i \leq m}\set{\frac{N_i^{'}-M_i^{'}}{\choose{m}{i}}}=\overleftarrow{\delta}_j.
\end{align*}

Note that $B_j=1-C_j$ for $d\leq j \leq d+k-1$.
Therefore, routing 1 is at least as reliable as routing 2 for
\begin{align*}
p=1-q
\geq& 1-\min\set{0.5,\max\limits_{d\leq j\leq d+k-1}B_j} \\
=&\max\set{0.5, \min\limits_{d\leq j\leq d+k-1}C_j}.
\end{align*}
This completes the proof.
\section{Lightpath Routing ILP to Minimize Minimum Cross Layer Spanning Tree (MCLST) Size}
\label{app:algorithm_high_regime}
As discussed in~\secref{high-regime}, lightpath routings with smaller MCLST size will be
more reliable in the high failure probability regime. In this section, we present an ILP
for the lightpath routing formulation that minimizes the MCLST.
that are optimized for the high failure probability regime.
We first define the following variables:
\begin{itemize}
\item $\set{f^{st}_{ij}| (s,t)\in E_L, (i,j)\in E_P}$: Flow variables representing the lightpath routing.
\item $\set{y_{ij}| (i,j)\in E_P}$: 1 if fiber $(i,j)$ survives, 0 otherwise.
\item $\set{z^{st}| (s,t)\in E_L}$: 1 if lightpath $(s,t)$ survives, 0 otherwise.
\item $\set{x^{st}| (s,t)\in E_L}$: Flow variables on the logical topology.
\end{itemize}
\begin{eqnarray}
\mathsf{MCLST}:\quad\mbox{Minimize\ } \sum_{(i,j)\in E_P} y_{ij}, \quad\mbox{subject to:} \nonumber \\
\label{con:logical_flow}
\sum\limits_{t\in V_L}x^{st}-\sum\limits_{t\in V_L}x^{ts}=\left\{
\begin{array}{l l}
|V_L|-1,\quad\mbox{if $s=0$} \\
-1, \quad\mbox{if $s\in V_L-\set{0}$}
\end{array}
\right.  \\
(V_L-1)\cdot z^{st}\geq x^{st}, \quad\forall (s,t)\in E_L\label{con:logical_link_survival}\\
\label{con:spanning_tree_link}
y_{ij}\geq z^{st} + f_{ij}^{st} - 1 \quad\forall (s,t)\in E_L, \forall (i,j)\in E_P\\
\{(i,j):f_{ij}^{st}=1\}\,\textrm{\small forms an $(s,t)$-path in $G_P$,}\,\, \forall (s,t)\in E_L \nonumber\\
0\leq y_{ij}\leq 1;\quad 0\leq x^{st};\quad  z_{ij}, f_{ij}^{st}\in\{0,1\} \nonumber
\end{eqnarray}

The variables $x^{st}$ represent a flow on the logical topology where 1 unit of flow is sent from logical node 0 to every other logical node, as described by~\conref{logical_flow}.~\conref{logical_link_survival} requires these flows to be carried only on the surviving logical links, which implies that the surviving links form a connected logical subgraph.~\conref{spanning_tree_link} ensures the survival of physical links that are used by any surviving logical links. Since the objective function minimizes $\sum\limits_{(i,j)\in E_P} y_{ij}$, the optimal solution will represent a minimum set of physical links whose survival will allow the logical link to be connected. 

Therefore, the set of physical links $(i,j)$ with $y_{ij}=1$ forms a cross-layer spanning tree. As a result, the optimal solution to the above ILP yields a lightpath routing that minimizes the size of the MCLST.


\bibliographystyle{IEEEtran}
\bibliography{main}

\newcommand{\noopsort}[1]{} \newcommand{\printfirst}[2]{#1}
  \newcommand{\singleletter}[1]{#1} \newcommand{\switchargs}[2]{#2#1}
\begin{thebibliography}{10}
\providecommand{\url}[1]{#1}
\csname url@samestyle\endcsname
\providecommand{\newblock}{\relax}
\providecommand{\bibinfo}[2]{#2}
\providecommand{\BIBentrySTDinterwordspacing}{\spaceskip=0pt\relax}
\providecommand{\BIBentryALTinterwordstretchfactor}{4}
\providecommand{\BIBentryALTinterwordspacing}{\spaceskip=\fontdimen2\font plus
\BIBentryALTinterwordstretchfactor\fontdimen3\font minus
  \fontdimen4\font\relax}
\providecommand{\BIBforeignlanguage}[2]{{%
\expandafter\ifx\csname l@#1\endcsname\relax
\typeout{** WARNING: IEEEtran.bst: No hyphenation pattern has been}%
\typeout{** loaded for the language `#1'. Using the pattern for}%
\typeout{** the default language instead.}%
\else
\language=\csname l@#1\endcsname
\fi
#2}}
\providecommand{\BIBdecl}{\relax}
\BIBdecl

\bibitem{zhang:review}
J.~Zhang and B.~Mukherjee, ``A review of fault management in {WDM} mesh
  networks: Basic concepts and research challenges,'' \emph{IEEE Network}, pp.
  41--48, Mar 2004.

\bibitem{zhang:service}
J.~Zhang, K.~Zhu, H.~Zang, and B.~Mukherjee, ``Service provisioning to provide
  per-connection-based availability guarantee in {WDM} mesh networks,'' in
  \emph{OFC}, Atlanta, GA, Mar 2003.

\bibitem{tornatore:availability}
M.~Tornatore, G.~Maier, and A.~Pattavina, ``Availability design of optical
  transport networks,'' \emph{IEEE Journal on Selected Areas in
  Communications}, vol.~23, no.~8, pp. 1520--1532, August 2005.

\bibitem{segovia:topology}
J.~{Segovia}, E.~{Calle}, P.~{Vila}, J.~{Marzo}, and J.~{Tapolcai},
  ``Topology-focused availability analysis of basic protection schemes in
  optical transport networks,'' \emph{J. of Opt. Net.}, vol.~7, pp. 351--364,
  Mar. 2008.

\bibitem{boesch:existence}
F.~Boesch, X.~Li, and C.~Suffel, ``On the existence of uniformly optimally
  reliable networks,'' \emph{Networks}, vol.~21, pp. 181--194, 1991.

\bibitem{wang:proof}
G.~Wang, ``A proof of {Boesch's} conjecture,'' \emph{Networks}, vol.~24, pp.
  277--284, 1994.

\bibitem{myrvold:uniformly}
W.~Myrvold, K.~Cheung, L.~Page, and J.~Perry, ``Uniformly-most reliable
  networks do not always exist,'' \emph{Networks}, vol.~21, pp. 417--419, 1991.

\bibitem{bauer:validity}
D.~Bauer, F.~Boesch, C.~Suffel, and {R. Van Slyke}, ``On the validity of a
  reduction of reliable network design to a graph extremal problem,''
  \emph{IEEE Trans. on Circ. and Sys.}, vol.~34, no.~12, pp. 1579--1581,
  December 1987.

\bibitem{myrvold:reliable}
W.~Myrvold, ``Reliable network synthesis: Some recent developments,'' in
  \emph{International Conference on Graph Theory, Combinatorics, Algorithms,
  and Applications, volume II}, 1996, pp. 650--660.

\bibitem{ath:some}
Y.~Ath and M.~Sobel, ``Some conjectured uniformly optimal reliable networks,''
  \emph{Probab. Eng. Inf. Sci.}, vol.~14, no.~3, pp. 375--383, 2000.

\bibitem{ath:counterexamples}
------, ``Counterexamples to conjectures for uniformly optimally reliable
  graphs,'' \emph{Probab. Eng. Inf. Sci.}, vol.~14, no.~2, pp. 173--177, 2000.

\bibitem{boesch:unreliability}
F.~Boesch, ``On unreliability polynomials and graph connectivity in reliable
  network synthesis,'' \emph{Journal of Graph Theory}, vol.~10, no.~3, pp.
  339--352, 1986.

\bibitem{kurant:survey}
M.~Kurant, H.~X. Nguyen, and P.~Thiran, ``Survey on dependable {IP} over fiber
  networks,'' \emph{LNCS}, vol. 4028, pp. 55--81, November 2006.

\bibitem{ArmitageINFOCOM97}
J.~Armitage, O.~Crochat, and J.-Y.~L. Boudec, ``Design of a survivable {WDM}
  photonic network,'' in \emph{IEEE INFOCOM}, April 1997.

\bibitem{CrochatJSAC98}
O.~Crochat and J.-Y.~L. Boudec, ``Design protection for {WDM} optical
  networks,'' \emph{IEEE JSAC}, vol.~16, no.~7, Sep. 1998.

\bibitem{sasakichingfong02}
Q.~Deng, G.~Sasaki, and C.-F. Su, ``Survivable {IP} over {WDM}: a mathematical
  programming problem formulation,'' in \emph{Proc 40th Allerton Conference on
  Communication, Control and Computing}, Oct. 2002.

\bibitem{kurant:survivable}
M.~Kurant and P.~Thiran, ``Survivable routing of mesh topologies in
  {IP-over-WDM} networks by recursive graph contraction,'' \emph{IEEE JSAC},
  vol.~25, no.~5, pp. 922--933, June 2007.

\bibitem{SubramaniamChoiSurvivableEmbedding}
H.~Lee, H.~Choi, S.~Subramaniam, and H.-A. Choi, ``Survivable embedding of
  logical topologies in {WDM} ring networks,'' \emph{Inf. Sci. Inf. Comput.
  Sci.}, vol. 149, no. 1-3, pp. 151--160, 2003.

\bibitem{modiano02survivable}
E.~Modiano and A.~Narula-Tam, ``Survivable lightpath routing: A new approach to
  the design of {WDM}-based networks,'' \emph{IEEE Journal on Selected Areas in
  Communications}, vol.~20, no.~4, May 2002.

\bibitem{binhaoISCC}
A.~Sen, B.~Hao, and B.~H. Shen, ``Survivable routing in {WDM} networks,'' in
  \emph{Proceedings of the Seventh International Symposium on Computers and
  Communications}, Washington, DC, USA, 2002.

\bibitem{binhaoHPSR03}
------, ``Minimum cost ring survivability in {WDM} networks,'' in
  \emph{Workshop on High Performance Switching and Routing}, 2003, pp.
  183--188.

\bibitem{ramamurthyBroadnets04}
A.~Todimala and B.~Ramamurthy, ``Survivable virtual topology routing under
  {Shared Risk Link Groups} in {WDM} networks,'' in \emph{BROADNETS},
  Washington, DC, USA, 2004.

\bibitem{lee:reliability}
K.~Lee, H.-W. Lee, and E.~Modiano, ``Reliability in layered networks with
  random link failures,'' \emph{IEEE/ACM Transactions on Networking}, vol.~19,
  no.~6, pp. 1835 --1848, dec. 2011.

\bibitem{to:unavailability}
M.~To and P.~Neusy, ``Unavailability analysis of long-haul networks,''
  \emph{IEEE JSAC}, vol.~12, no.~1, pp. 100--109, 1994.

\bibitem{Wu09}
W.~Wu, B.~Moran, J.~Manton, and M.~Zukerman, ``Topology design of undersea
  cables considering survivability under major disasters,'' in \emph{WAINA},
  May 2009.

\bibitem{Neu11}
S.~Neumayer, G.~Zussman, R.~Cohen, and E.~Modiano, ``Assessing the
  vulnerability of the fiber infrastructure to disasters,'' \emph{Networking,
  IEEE/ACM Transactions on}, vol.~19, no.~6, pp. 1610 --1623, dec. 2011.

\bibitem{Fos04}
{J. S. Foster Jr. et al.}, ``Report of the commission to assess the threat to
  the united states from electromagnetic pulse (emp) attack, vol. i: Executive
  report,'' Apr. 2004.

\bibitem{Wil04}
\BIBentryALTinterwordspacing
C.~Wilson. (2004) High altitude electromagnetic pulse (hemp) and high power
  microwave (hpm) devices: Threat assessments. [Online]. Available:
  \url{http://www.fas.org/man/crs/RL32544.pdf}
\BIBentrySTDinterwordspacing

\bibitem{Rad07}
\BIBentryALTinterwordspacing
W.~Radasky. (2007) High-altitude electromagnetic pulse (hemp): A threat to our
  way of life. [Online]. Available:
  \url{http://www.todaysengineer.org/2007/Sep/HEMP.asp}
\BIBentrySTDinterwordspacing

\bibitem{Fed10}
``Detailed technical report on emp and severe solar flare threats to the u.s.
  power grid,'' Oct. 2010.

\bibitem{rosato2008}
V.~Rosato, L.~Issacharoff, F.~Tiriticco, S.~Meloni, S.~Porcellinis, and
  R.~Setola, ``Modelling interdependent infrastructures using interacting
  dynamical models,'' \emph{International Journal of Critical Infrastructures},
  vol.~4, no.~1, pp. 63--79, 2008.

\bibitem{lee:crosslayer}
K.~Lee, E.~Modiano, and H.-W. Lee, ``Cross-layer survivability in {WDM}-based
  networks,'' \emph{IEEE/ACM Transactions on Networking}, vol.~19, no.~4, pp.
  1000 --1013, aug. 2011.

\bibitem{kayi:thesis}
K.~Lee, ``Survivability in layered networks,'' Ph.D. dissertation,
  Massachusetts Institute of Technology, Cambridge, MA, 2011.

\bibitem{jasonjue:min_color_path_infocom}
S.~Yuan, S.~Varma, and J.~P. Jue, ``Minimum-color path problems for reliability
  in mesh networks,'' in \emph{INFOCOM}, 2005, pp. 2658--2669.

\bibitem{yen:k_shortest}
J.~Y. Yen, ``Finding the $k$ shortest loopless paths in a network,''
  \emph{Management Science}, vol.~17, no.~11, pp. 712--716, 1971.

\bibitem{AndrasAugmentation}
A.~Frank, ``Connectivity augmentation of networks: structures and algorithms,''
  \emph{Mathematical Programming}, vol.~84, pp. 439--441, 1999.

\bibitem{CaiSunAugmentation}
G.-R. Cai and Y.-G. Sun, ``The minimum augmentation of any graph to a
  k-edge-connected graph,'' \emph{Networks}, vol.~19, no.~1, pp. 151--172,
  January 1989.

\bibitem{Hsu93graphaugmentation}
T.-S. Hsu, T.~Suel, R.~Tamassia, and H.~Y.~D. The, ``Graph augmentation and
  related problems: Theory and practice,'' 1993.

\bibitem{112607}
T.~Watanabe, Y.~Higashi, and A.~Nakamura, ``An approach to robust network
  construction from graph augmentation problems,'' vol.~4, May 1990, pp.
  2861--2864.

\bibitem{Jackson2000185}
B.~Jackson and T.~Jord$\mathrm{\acute{a}}$n, ``Connectivity augmentation of
  graphs,'' \emph{Electronic Notes in Discrete Mathematics}, vol.~5, pp.
  185--188, 2000, 6th International Conference on Graph Theory.

\bibitem{qwest}
``Qwest {C}ommunications,''
  \url{https://www.qwest.com/business/resource-center/network-maps/list.html}.

\bibitem{We2004}
G.~Weichenberg, V.~Chan, and M.~Medard, ``High-reliability topological
  architectures for networks under stress,'' \emph{IEEE Journal on Selected
  Areas in Communications}, vol.~22, no.~9, pp. 1830 -- 1845, nov. 2004.

\end{thebibliography}


\begin{IEEEbiography}
{Hyang-Won Lee}(S'02-M'08) received the BS, MS, and PhD degrees in all electrical engineering and computer science from the Korea Advanced Institute of Science and Technology, Daejeon, Korea, in 2001, 2003, and 2007, respectively. He was a postdoctoral research associate at the Massachusetts Institute of Technology during 2007-2011, and a research assistant professor at KAIST between 2011 and 2012. He is currently an assistant professor in the Department of Internet and Multimedia Engineering, Konkuk University, Seoul, Republic of Korea. His research interests are in the areas of congestion control, stochastic network control, robust optimization and reliable network design.
\end{IEEEbiography}\vspace{-0.8cm}

\begin{IEEEbiography}
{Kayi Lee}(S'99-M'05) received the B.S. degree in Computer Science and Mathematics from the Massachusetts Institute of Technology (MIT), Cambridge, in 1999 and his M.Eng., and PhD degrees in Computer Science from MIT in 2000 and 2011 respectively. He is currently a Software Engineer at Google Inc., Cambridge, MA. His research interests include routing and network survivability.
\end{IEEEbiography}\vspace{-0.8cm}

\begin{IEEEbiography}
{Eytan Modiano}(S'90-M'93-SM'00-F'12) received the BS degree in electrical engineering and computer science from the University of Connecticut at Storrs in 1986 and the MS and PhD degrees, both in electrical engineering, from the University of Maryland, College Park, in 1989 and 1992, respectively. He was a US Naval Research Laboratory fellow between 1987 and 1992 and a US National Research Council postdoctoral fellow during 1992-1993. Between 1993 and 1999, he was with the Massachusetts Institute of Technology (MIT) Lincoln Laboratory where he was the project leader for MIT Lincoln Laboratory’s Next Generation Internet (NGI) project. Since 1999, he has been on the faculty at MIT, where he is presently a professor. His research is on communication networks and protocols with emphasis on satellite, wireless, and optical networks. He is currently an associate editor for the IEEE/ACM Transactions on Networking and for The International Journal of Satellite Communications and Networks. He had served as associate editor for the IEEE Transactions on Information Theory (communication networks area) and as a guest editor for the IEEE JSAC special issue on WDM network architectures, the Computer Networks journal special issue on broadband internet access, the Journal of Communications and Networks special issue on wireless ad-hoc networks, and for the IEEE Journal of Lightwave Technology special issue on optical networks. He was the technical program cochair for Wiopt 2006, IEEE INFOCOM 2007, and ACM MobiHoc 2007. He is an associate fellow of the AIAA and a fellow of the IEEE.
\end{IEEEbiography}

\end{document}